\documentclass[a4paper,USenglish,cleveref, autoref, thm-restate]{lipics-v2021}
 \newif\iflong
\newif\ifshort
 
\longtrue

\iflong
\else
\shorttrue
\fi

 \newcommand{\lv}[1]{}

\usepackage{verbatim,subcaption}
\usepackage{xcolor}
\usepackage[textsize=footnotesize]{todonotes}
\usepackage{scalerel}
\usepackage{graphicx}
\usepackage{enumerate}
\usepackage{graphics}
\usepackage{xspace}
\usepackage{mathabx}
  
\usepackage{caption}
\usepackage{subcaption}
\usepackage{epsfig}
\usepackage{arcs}
\usepackage[T1]{fontenc}
\usepackage{alltt}
\usepackage{amssymb,amsfonts,epsf,url}
\usepackage{graphicx}
\usepackage{pstricks,pstricks-add}
\usepackage{pst-node}
\usepackage{pst-coil}
\usepackage{amsmath}
\usepackage{amsthm}
\usepackage{verbatim}
\usepackage[textsize=footnotesize]{todonotes}
\usepackage{xcolor}
\usepackage[textsize=footnotesize]{todonotes}
\usepackage{fancyhdr}
\usetikzlibrary{arrows}

\newcommand{\translate}{\textup{\textsf{translate}}}
\newcommand{\trace}{\textup{\textsf{trace}}}

\newcommand{\spcmp}{\textsc{\textup{Rect-MP}}\xspace} 
\newcommand{\ppcmp}{\textsc{\textup{Rect-CMP}}\xspace} 
\newcommand{\prect}{\textsc{\textup{Rect-CMP}}\xspace}
\newcommand{\srect}{\textsc{\textup{Rect-MP}}\xspace}

\def\msquare{\mathord{\scalebox{0.7}[0.7]{\scalerel*{\Box}{\strut}}}}
 
\theoremstyle{plain}

\newcommand{\bigoh}{\mathcal{O}}

\newtheorem{fact}[theorem]{Fact}

\newlength{\alginputwidth}

\newcommand{\Oh}{{\mathcal O}}
\newcommand{\nat}{\mathbb{N}}

\newcommand{\PSPACE}{\mbox{\sf PSPACE}}
\newcommand{\NP}{\mbox{{\sf NP}}}

\newcommand{\FPT}{\mbox{{\sf FPT}}}

\newcommand{\R}{\mathcal{R}}
\newcommand{\I}{\mathcal{I}}

\newcommand{\SSS}{\mathcal{S}}

\newcommand{\VVV}{\mathcal{V}}

\newcommand{\LLL}{\mathcal{L}}

\newcommand{\normalproblem}[3]{\noindent  
{\sc #1}
\\
{\bf Given:} #2\\
{\bf Question:} #3
 
\medskip
}

\newcommand {\mm}[1] {\ifmmode{#1}\else{\mbox{\(#1\)}}\fi}

\newcommand{\Rspace}        {\mm{\mathbb{R}}}

\newcommand{\eps}        {\varepsilon}

\newcommand{\problem}{Rectangular Motion Planning }
\newcommand{\region}{feasibility region}
\newcommand{\reach}{V}

 \title{On the Parameterized Complexity of Motion Planning for Rectangular Robots}  
\titlerunning{Motion Planning for Rectangular Robots}

\author{Iyad Kanj}{School of Computing, DePaul University, Chicago, USA}{ikanj@depaul.edu}{0000-0003-1698-8829}{}
\author{Salman Parsa}{School of Computing, DePaul University, Chicago, USA}{s.parsa@depaul.edu}{0000-0002-8179-9322}{}

\authorrunning{I.\ Kanj, S.\ Parsa} %

\Copyright{Iyad Kanj, Salman Parsa} %



\category{} 

\relatedversion{} 

\acknowledgements{}

\iflong
\nolinenumbers 
\fi

\ifshort

\EventEditors{Wolfgang Mulzer and Jeff M. Philips} \EventNoEds{2} \EventLongTitle{40th International Symposium on Computational Geometry (SoCG 2024)} \EventShortTitle{SoCG 2024} \EventAcronym{SoCG} \EventYear{2024} \EventDate{June 11--14, 2024} \EventLocation{Athens, Greece} \EventLogo{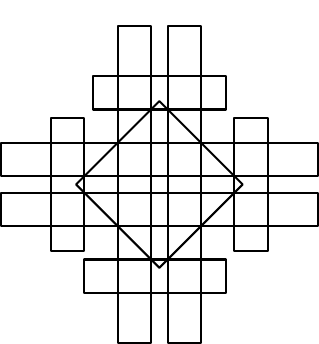} \SeriesVolume{} \ArticleNo{}  
\fi
\begin{document}

\maketitle

\begin{abstract} 
We study computationally-hard fundamental motion planning problems where the goal is to translate $k$ axis-aligned rectangular robots from their initial positions to their final positions without collision, and with the minimum number of translation moves. Our aim is to understand the interplay between the number of robots and the geometric complexity of the input instance measured by the input size, which is the number of bits needed to encode the coordinates of the rectangles' vertices. We focus on axis-aligned translations, and more generally, translations restricted to a given set of directions, and we study the two settings where the robots move in the free plane, and where they are confined to a bounding box. We also consider two modes of motion: serial and parallel. We obtain fixed-parameter tractable (\FPT) algorithms parameterized by $k$ for all the settings under consideration.

In the case where the robots move serially (i.e., one in each time step) and axis-aligned, we prove a structural result stating that every problem instance admits an optimal solution in which the moves are along a grid, whose size is a function of $k$, that can be defined based on the input instance. This structural result implies that the problem is fixed-parameter tractable parameterized by $k$. 

We also consider the case in which the robots move in parallel (i.e., multiple robots can move during the same time step), and which falls under the category of Coordinated Motion Planning problems. Our techniques for the axis-aligned motion here differ from those for the case of serial motion. We employ a search tree approach and perform a careful examination of the relative geometric positions of the robots that allow us to reduce the problem to \FPT-many Linear Programming instances, thus obtaining an \FPT{} algorithm.

Finally, we show that, when the robots move in the free plane, the \FPT{} results for the serial motion case carry over to the case where the translations are restricted to any given set of directions.

\end{abstract}

\section{Introduction}
\label{sec:intro}

\subsection{Motivation}
We study the parameterized complexity of computationally-hard fundamental motion planning problems where the goal is to translate $k$ axis-aligned rectangular robots from their initial positions to their final positions without collision, and with the minimum number of translation moves. The parameter under consideration is the number $k$ of robots, and the input length $N$ is the number of bits needed to encode the coordinates of the vertices of the rectangles. We point out that, in our study, we deviate from using the real RAM model~\cite{shamos}, which assumes that arithmetic operations over the reals can be performed in constant time, and use the Turing machine model instead. We believe that our study is more faithful to the geometric setting under consideration than the real RAM model. The input length $N$ can be much larger than the number $k$ of rectangles, and hence, our study of the parameterized complexity of the problem, which aims at investigating whether the problem admits algorithms whose running time is polynomially-dependent on the input size, is both meaningful and significant in such settings.

We study two settings where the robots move in the free plane, and where they are confined to a bounding box. We also consider two modes of motion: serial and parallel. Problems  with the latter motion mode fall under the category of Coordinated Motion Planning problems.  We point out that the problems under consideration have close connections to well-studied motion planning and reconfiguration problems, including the famous \NP-complete $(n^2-1)$-puzzle~\cite{nphard2,nphard1} and the \PSPACE-hard warehouseman's problem~\cite{hopcroft1984} where the movement directions are limited, among many others. Moreover, the Coordinated Motion Planning for robots moving on a rectangular grid featured as the SoCG 2021 Challenge~\cite{socg2021}.

For most natural geometric (or continuous) motion planning problems, pertaining to the motion of well-defined geometric shapes in an environment with/without polygonal obstacles,
the feasibility of an instance of the problem can be formulated as a statement in the first-order theory of the reals. Therefore, it is decidable using Tarski's method in time that is polynomially-dependent on the input length, and exponentially-dependent on the number of variables, number of polynomials and the highest degree over all the polynomials in the statement (see~\cite{renegar1,renegar2, renegar3,schwartzsharir2}). When the parameter is the number $k$ of robots, if an upper bound in $k$ on the number of moves in a solution exists, then the existence of a solution can be decided in \FPT{}-time using the above general machinery. However, this approach is non-constructive, and we might not be able to extract in \FPT-time a solution to a feasible instance, as the only information we have about the solution is that it is algebraic. 

There has been very little work on the parameterized complexity of these fundamental geometric motion planning problems, and our understanding of their parameterized complexity is lacking. Most of the early work (e.g., see~\cite{alagar,schwartzsharir3}) on such problems have resulted in algorithms for deciding \emph{only} the feasibility of the instance and whose running time takes the form $\Oh(n^{f(k)})$, where $n$ is the number of edges/walls composing the polygonal obstacles in the environment. Therefore, it is natural to investigate the parameterized complexity of the more practical variants of the problems, where one seeks a solution that meets a given upper bound on the number of robot moves or an optimal solution w.r.t.~the number of robot moves, which remained unanswered by the earlier works. 

The goal of this paper is to shed light on the parameterized complexity of these motion planning problems by considering the very natural setting of axis-aligned translations (i.e., horizontal and vertical), and more generally, translations restricted to a given (or a fixed) set of directions. We aim to understand the interplay between the number of robots and the complexity of the input instance (i.e., the input size). Our results settle the parameterized complexity of most of the studied problem variants by showing that they are \FPT{}.

\subsection{Related Work}
There has been a lot of work, dating back to the 1980's, on the motion planning of geometric shapes  (e.g., disks, rectangles, polygons) in the Euclidean plane (with or without obstacles), motivated by their applications in robotics. In this setting, robots may move along continuous curves. The problem is very hard, and most of the work focused on the feasibility of the problem for various shapes and environment settings (disks, rectangles, obstacle-free environment, environment with polygonal obstacles, etc.).

The early works by Schwartz and Sharir~\cite{schwartzsharir3,schwartzsharir1, schwartzsharir2} showed that deciding the feasibility of an instance of the problem for two disks in a region bounded by $n$ ``walls'' can be done in time $\Oh(n^3)$~\cite{schwartzsharir3}; they mentioned that their result can be generalized to any number, $k$, of disks to yield an $\Oh(n^{h(k)})$-time algorithm, for some function $h$ of $k$. When studying feasibility, the moves can be assumed to be performed serially, and a move may extend over any Euclidean length. Ramanathan and Alagar~\cite{alagar} improved the result of Schwartz and Sharir~\cite{schwartzsharir3}  to $\Oh(n^{k})$, conjecturing that this running time is asymptotically optimal. The feasibility of the coordinated motion planning of rectangular robots confined to a bounding box was shown to be \PSPACE-hard~\cite{hopcroft1984,hopcroft1986}. The problem of moving disks among polygonal obstacles in the plane was shown be strongly \NP-hard~\cite{spirakis}; when the shapes are unit squares, Solovey and Halprin~\cite{solovey} showed the problem to be \PSPACE-hard.

Dumitrescu and Jiang~\cite{dumitrescu2013}
studied the problem of moving unit disks in an obstacle-free environment. They consider two types of moves: translation (i.e., a linear move) and sliding (i.e., a move along a continuous curve). In a single step, a unit disk may move any distance either along a line (translation) or a curve (sliding) provided that it does not collide with another disk. They showed that deciding whether the disks can reach their destinations within $r \in \mathbb{N}$ moves is \NP-hard, for either of the two movement types. 
Constant-ratio approximation algorithms for the coordinated motion planning of unit disks in the plane, under some separation condition, where given in~\cite{demaine}. For further work on the motion planning of disks, we refer to the survey of Dumitrescu~\cite{survey}.

The problem of moving unit disks in the plane is related to the problem of reconfiguring/moving coins, which has also been studied and shown to be \NP-hard~\cite{coins}.
Moreover, there has been work on the continuous collision-free motion of a constant number of rectangles in the plane, from their initial positions to their final positions, with the goal of optimizing the total Euclidean covered length; we refer to~\cite{agarwal2023,esteban2023} for some of the most recent works on this topic.

Perhaps the most relevant, but orthogonal, work to ours, in the sense that it pertains to studying the parameterized complexity of translating rectangles, is the paper of Fernau et al.~\cite{fernaucccg}. In~\cite{fernaucccg}, they considered a geometric variant of the \PSPACE-complete Rush-Hour problem, which itself was shown to be \PSPACE-complete~\cite{flake2002}. In this variant, cars are represented by rectangles confined to a bounding box, and cars move serially. Each car can either move horizontally or vertically (or not move at all, i.e., is an obstacle), but never both during its whole motion; that is, each car slides on a horizontal track, or a vertical track. The goal is to navigate each car to its destination and a designated car to a designated rectangle in the box (whose corner coincides with the origin). They showed that the problem is \FPT{} when parameterized by either the number of cars or the number of moves.  

Finally, we mention that Eiben et al.~\cite{egksocg2023} studied the parameterized complexity of Coordinated Motion Planning in the combinatorial setting where the robots move on a rectangular grid. Among other results, they presented \FPT{} algorithms,  parameterized by the number of robots, for each of the two objective targets of minimizing the makespan and the total travel length~\cite{egksocg2023}.

\subsection{Contributions}
We present fixed-parameter algorithms parameterized by the number $k$ of (rectangular axis-aligned) robots for most of the problem variants and settings under consideration.  A byproduct of our results for the problems under consideration is that rational instances of these problems admit rational solutions.

\begin{itemize}
 \item[(i)] We give an \FPT-algorithm for the axis-aligned serial motion in the free plane. Our proof relies on a structural result stating that every problem instance admits an optimal solution in which the moves are along a grid that can be defined based on the input instance. This structural result, combined with an upper bound of $4k$ that we prove on the number of moves in the solution to a feasible instance, implies that the problem is solvable in time $\Oh^*(k^{16k} \cdot 2^{20k^2+8k})$, and hence is \FPT.  
\end{itemize}

The structural result does not apply when the translations are not axis-aligned. To obtain \FPT{} results for these cases, we employ a search-tree approach, and perform a careful examination of the relative geometric positions of the robots, that allow us to reduce the problem to \FPT-many Linear Programming instances. 
 
\begin{itemize}
 \item[(ii)]  We show that the problem for serial motion in the free plane 
 for any fixed-cardinality given set $\VVV$ of directions (i.e., part of the input) is solvable in time $\Oh^*((4^{32k} \cdot k \cdot |\VVV|)^{4k})$. A byproduct of this \FPT{} algorithm is that the problem is in $\NP$, a result that -- up to the authors' knowledge -- was not known nor is obvious. We complement this result by showing that the aforementioned problem for any fixed set $\VVV$ of directions that contains at least two nonparallel directions (which includes the case where the motion is axis-aligned) is \NP-hard, thus concluding that the problem is \NP-complete.  

 \item[(iii)] We give an \FPT{} algorithm for the problem where the serial motion is axis-aligned and confined to a bounding box, which was shown to be \PSPACE-hard in~\cite{hopcroft1984}. This result is obtained after proving an upper bound of $2k \cdot 5^{k(k-1)}$ on the number of moves in a feasible instance of the problem.
 

 The approach used in (ii) and (iii) does not extend seamlessly to the case of coordinated motion (i.e., when robots move in parallel), as modelling collision in the case of parallel motion becomes more involved. Nevertheless, by a more careful enumeration and examination of the relative geometric positions of the robots, we give:

 \item[(iv)] An \FPT{} algorithm for the axis-aligned coordinated motion planing in the free plane that runs in $\Oh^*(5^{2k^3}\cdot 8^{4k^2})$ time, and an \FPT{} algorithm for the axis-aligned coordinated motion planning confined to a bounding box that runs in time $\Oh^*(5^{k^2}\cdot 8^{2k^2\cdot5^{k^2}} \cdot 5^{k^3\cdot 5^{k^2}})$. The \FPT{} algorithm for the former problem implies its membership in \NP. 
\end{itemize}

\ifshort
The details and the proofs of some of the results appear in the full version of the paper, which is included as an appendix.
\fi
  
\section{Preliminaries and Problem Definition}
\label{sec:prelim}

We denote by $[k]$ the set $\{1, \ldots, k\}$. Let $\R=\{R_i \mid i \in [k]\}$ be a set of axis-aligned rectangular robots. For $R_i \in \R$, we denote by $x(R_i)$ and $y(R_i)$ the horizontal and vertical dimensions of $R_i$, respectively. We will refer to a robot by its identifying name (e.g., $R_i$), which determines its location in the schedule at any time step, even though, when it is clear from the context, we will identify the robot with the rectangle it represents/occupies at a certain time step.

A \emph{translation move}, or a \emph{move}, for a robot $R_i \in \R$ w.r.t.~a direction $\overrightarrow{v}$, is a translation of $R_i$ by a vector $\alpha \cdot \overrightarrow{v}$ for some $\alpha  > 0$. For a vector $\vec{u}$,  \translate($R_i$, $\overrightarrow{u}$) denotes the axis-aligned rectangle resulting from the translation of $R_i$ by vector $\overrightarrow{u}$.
We denote by \emph{axis-aligned motion} the translation motion with respect to the set of four directions $\VVV=\{\overrightarrow{H}^{-}, \overrightarrow{H}^{+}, \overrightarrow{V}^{-}, \overrightarrow{V}^{+}\}$, which are the negative and positive unit vectors of the $x$- and $y$-axis, respectively.

In this paper, we consider two types of moves: \emph{serial} and \emph{parallel}, where the former type corresponds to the robots moving one at a time (i.e., a robot must finish its move before the next starts), and the latter type corresponds to (possibly) multiple robots moving simultaneously. We now define  collision for the two types of motion.

For a robot $R_i$ that is translated by a vector $\overrightarrow{v}$, we say that $R_i$ \emph{collides} with a stationary robot $R_j \neq R_i$, if there exists $0 \leq x \leq 1$ such that $R_j$ and  \translate($R_i, x\cdot \overrightarrow{v}$) intersect in their interior. For two distinct robots $R_i$ and $R_j$ that are simultaneously translated by vectors $\overrightarrow{v}$ and $\overrightarrow{u}$, respectively, we say that $R_i$ and $R_j$ \emph{collide} if there exists $0 \leq x \leq 1$ such that  \translate($R_i, x \cdot \overrightarrow{v}$) and  \translate($R_j, x \cdot \overrightarrow{u}$) intersect in their interior.

We think of $\R$ as a set of axis-aligned rectangular robots, where each robot is given by the rectangle of its starting position and the congruent rectangle of its desired final position. We assume that the starting rectangles (resp.~final destination rectangles) of the robots are pairwise non-overlapping (in their interiors). 
Let $\VVV=\{\overrightarrow{\theta_1}, \ldots, \overrightarrow{\theta_c}\}$, where $c \in \nat$, be a set of unit vectors. We assume that if a vector $\overrightarrow{\theta_i}$ is in $\VVV$ then the vector $-\overrightarrow{\theta_i}$ is also in $\VVV$. For a vector $\vec{v}$, we denote by $x(\vec{v})$ and $y(\vec{v})$ the $x$-component/coordinate (i.e., projection of $\vec{v}$ on the $x$-axis) and $y$-component/coordinate of $\vec{v}$, respectively.  

A \emph{valid serial schedule} (resp.~\emph{valid parallel schedule}) $\SSS$ for $\R$ w.r.t~$\VVV$ is a sequence of collision-free serial (resp.~parallel) moves, where each move in $\SSS$ is along a direction (resp.~a set of directions) in $\VVV$, and after all the moves in $\SSS$, each $R_i$ ends at its final destination, for $i \in [k]$. The \emph{length} $|\SSS|$ of the schedule is the number of moves in it. In this paper, we study the following problem:
 
\normalproblem{{\sc Rectangles Motion Planning} (\spcmp)}{A set of pairwise non-overlapping axis-aligned rectangular robots $\R=\{R_i \mid i \in [k]\}$ each given with its starting and final positions/rectangles; a set $\VVV$ of directions; $k, \ell \in \nat$.}{Is there a valid schedule for $\R$ w.r.t.~$\VVV$ of length at most $\ell$?}

We note that the time complexity for solving the above decision problem will be essentially the same (up to a polynomial factor) as that for solving its optimization version (where we seek to minimize $\ell$),  as we can binary-search for the length of an optimal schedule.

We also study The \textsc{Rectangles Coordinated Motion Planning} problem (\prect{}), which is defined analogously with the only difference that the moves could be performed in parallel. More specifically, the schedule of the robots consists of a sequence of moves, where in each move a subset $S$ of robots move simultaneously, along (possibly different) directions from $\VVV$, at the same speed provided that no two robots in $\R$ collide. The move ends when all the robots in $S$ reach their desired locations during that move; no new robots (i.e., not in $S$) can move during that time step.

 We focus on the restrictions of \spcmp{} and \ppcmp to instances in which the translations are axis-aligned,
 but we also extend our results to the case where the directions are part of the input (or are fixed). We also consider both settings where the rectangles move freely in the plane, and where their motion is confined to a bounding box. For a problem $P \in$ $\{\spcmp{},\ppcmp\}$, denote by 
 $\bigplus$-$P$ the restriction of $P$ to instances in which the translations are axis-aligned (i.e., $\VVV=\{\overrightarrow{H}^{-}, \overrightarrow{H}^{+}, \overrightarrow{V}^{-}, \overrightarrow{V}^{+}\}$), by $P$-$\msquare$ its restriction to instances in which the robots are confined to a bounding box (which we assume that it is given as part of the input instance), and by $\bigplus$-$P$-$\msquare$ the problem satisfying both constraints. For instance, $\bigplus$-\srect{}-$\msquare$ denotes the problem in which the motion mode is serial, the translations are axis-aligned, and the movement is confined to a bounding box.

In parameterized
complexity~\cite{CyganFKLMPPS15,DowneyFellows13}, the
running-time of an algorithm is studied with respect to a parameter
$k\in \nat$ and input size~$N$. The basic idea is to find a parameter
that describes the structure of the instance such that the
combinatorial explosion can be confined to this parameter. In this
respect, the most favorable complexity class is \FPT{}
(\textit{fixed-parameter tractable}) which contains all problems that
can be decided in time $f(k)\cdot
N^{\bigoh(1)}$, where $f$ is a computable function. Algorithms with
this running-time are called \emph{fixed-parameter algorithms}.  

The $\Oh^*()$ notation hides a polynomial function in the input size $N$, which is the length of the binary encoding of the instance.

\ifshort

\section{Upper Bounds on the Number of Moves}
\label{sec:upperbound}
In this section, we prove upper bounds -- w.r.t.~the number $k$ of robots -- on the number of moves in an optimal schedule for feasible instances of several of the problems under consideration in this paper. These upper bounds are crucial for obtaining the \FPT{} results.

\subsection{Motion in the Free Plane}
The upper bound in the case where the robots move in the free plane follow since, given at least two non-parallel directions, one can translate the rectangles, one by one, to very far and well-separated locations, and then reverse the process to bring them to their destinations:

\begin{proposition}[Appendix]
\label{prop:free}
Let $\I=(\R, \VVV, k, \ell)$ be an instance of \spcmp{} or \ppcmp. If $\VVV$ contains two non-parallel directions, then there is a schedule for $\I$ of length at most $4k$.
\end{proposition}

\subsection{Axis-Aligned Motion in a Bounding Box}

Let $\I=(\R, \VVV, k, \ell, \Gamma)$ be an instance of $\bigplus$-\srect{}-$\msquare$. Fix an ordering on the vertices of any rectangle (say the clockwise ordering, starting always from the top left vertex). For any two robots $R$ and $R'$, the \emph{relative order} of $R$ w.r.t.~$R'$ is the order in which the vertices of $R$, when considered in the prescribed order, appear relatively to the vertices of $R'$ (considered in the prescribed order as well), with respect to each of the $x$-axis and the $y$-axis.  
\begin{definition}
\label{def:configutration}
Fix an arbitrary ordering of the 2-sets of robots in $\R$. A \emph{configuration} of $\R$ is a sequence indicating, for each 2-set $\{R, R'\}$ of robots in $\R$, considered in the prescribed order, the relative order of $R$ with respect to $R'$. A \emph{realization} of a configuration $C$ is an embedding of the robots in $\R$ such that the relative order of any two robots in $\R$ conforms to that described by $C$ and the robots in the embedding are pairwise non-overlapping. 
\end{definition}

The following proposition shows that we can move between any two realizations of the same configuration using at most $2k$ translations:

\begin{proposition}[Appendix]\label{p:aaconfiguration} 
\label{lem:realizations} For any two realizations $\rho, \rho'$ of a configuration $C$, there is a sequence of at most $2k$ valid moves within the bounding box $\Gamma$ that translate the robots from their positions in $\rho$ to their positions in $\rho'$.
\end{proposition}
The above proposition implies that each configuration appears at most $2k$ times in an optimal schedule. By upper bounding the total number of distinct configurations, we get:
 
\begin{proposition}[Appendix]
\label{prop:boundedconfigurations}
Let $\I=(\R, \VVV, k, \ell, \Gamma)$ be a feasible instance of $\bigplus$-\srect{}-$\msquare$ or $\bigplus$-\prect{}-$\msquare$. Then there is a schedule for $\I$ of length at most $2k \cdot 5^{k(k-1)}$.  
\end{proposition}

We note that the obtained upper bounds on the length of the schedule of a feasible instance already suggest an approach for obtaining \FPT{} algorithms, namely that of enumerating all sequences of configurations and then checking their realizability and serializability. This approach results in much worse running times than the ones for the \FPT{} algorithms presented in this paper, is certainly not easier (as it involves checking realizability and serializability), and reveals less about the structure of the solutions than the presented results.    
\fi

\iflong
\section{Upper Bounds on the Number of Moves}
\label{sec:upperbound}
In this section, we prove upper bounds -- w.r.t.~the number of robots $k$ -- on the number of moves in an optimal schedule for feasible instances of several of the problems under consideration in this paper. Those upper bounds will be crucial for obtaining the \FPT{} results in later sections.

\subsection{Motion in the Free Plane}
The upper bound in the case where the robots move in the free plane is straightforward to obtain: 

\begin{proposition}
\label{prop:free}
Let $\I=(\R, \VVV, k, \ell)$ be an instance of \spcmp{} or \ppcmp. If $\VVV$ contains at least two non-parallel directions, then there is a schedule for $\I$ of length at most $4k$.

\end{proposition}

\begin{proof}
 
Choose two non-parallel directions $\overrightarrow{v}$ and $\overrightarrow{w}$ in $\VVV$, and assume that $\overrightarrow{v}$ has a positive slope and $\overrightarrow{w}$ has a negative slope. The arguments for the other cases are analogous.
An easy observation shows that there is a robot $R_i \in \R$ such that the bottom-left quadrant defined by its top-right vertex does not overlap with any other robots. We can use a single move in the direction of $\pm \overrightarrow{v}$ to
translate the center of $R_i$ to a point that is extremely far from the other robots. By repeating this argument, we can separate every pair of robots by any arbitrary distance of our choice using $k$ moves. We can also perform this operation starting from the end positions following the direction $\pm \overrightarrow{w}$ and we can place the one-to-last
position of $R_i$ and the second position of $R_i$, say on concentric circles $C_i$ of distances arbitrarily-large from each other, see Figure~\ref{fig:freeplaneupperbound}. 
Hence, each robot $R_i$ has now two positions that need to be connected. This is possible with $2$ moves per robot. In total, we incur at most $4k$ moves to translate the robots in $\R$ from their starting positions to their destinations.

Note that we do not care about the geometric complexity or the number of bits required to encode the
intermediate positions as we only need to show an upper bound on the number of moves.

\end{proof}

\begin{figure}
    \centering
    \includegraphics[scale=0.75]{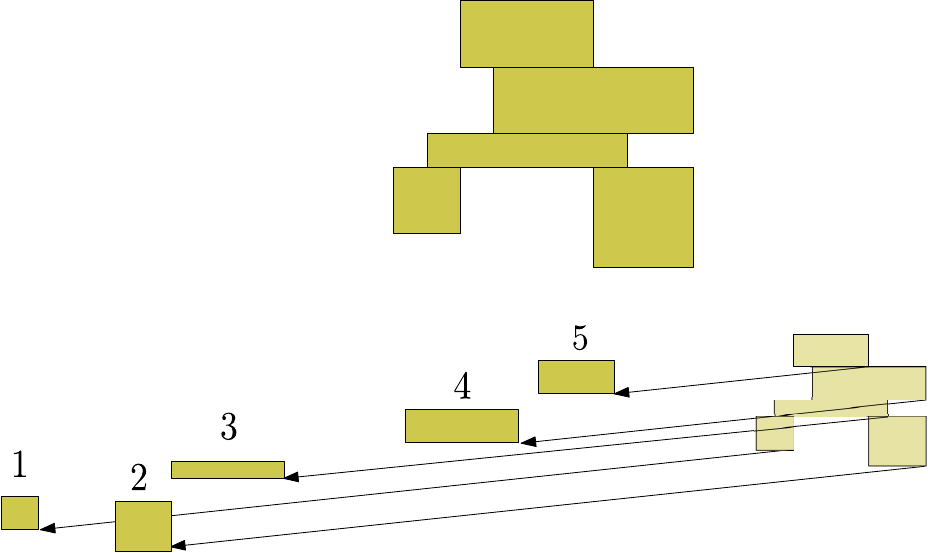}
    \caption{Top: robots at their starting position. Bottom: Robots (made smaller for visibility) can be separated by arbitrary large distance using a single direction. The numbers indicate the order in which the robots have to move. The figure is not at scale, as we might need very large distances to guarantee non-intersection in third and fourth moves of the robots.}
    \label{fig:freeplaneupperbound}
\end{figure}

\subsection{Axis-Aligned Motion in a Bounding Box}
In this section, we present upper bounds on the number of axis-aligned moves in the case where the (serial/parallel) axis-aligned motion takes place within an axis-aligned bounding box $\Gamma$, as opposed to the free plane. The crux of the proof is showing that if the instance is feasible (i.e., admits a schedule  within $\Gamma$), then it admits a schedule whose length is upper bounded by a computable function of the number $k$ of robots.   

Let $\I=(\R, \VVV, k, \ell, \Gamma)$ be an instance of $\bigplus$-\srect{}-$\msquare$. Fix an ordering on the vertices of any rectangle (say the clockwise ordering, starting always from the top left vertex). For any two robots $R$ and $R'$, the \emph{relative order} of $R$ w.r.t.~$R'$ is the order in which the vertices of $R$, when considered in the prescribed order, appear relatively to the vertices of $R'$ (considered in the prescribed order as well), with respect to each of the $x$- and $y$-axes. Note that the relative order of $R$ with respect to $R'$ determines, for each vertex $p$ of $R$, the relative position of the $x$-coordinate (resp.~$y$-coordinate) of $p$ with respect to the $x$-coordinate (resp.~$y$-coordinate) of each vertex of $R'$ (e.g, with respect to $\overrightarrow{x'x}$, the relative position of a vertex $p$ of $R$ w.r.t.~to a vertex $q$ of $R'$ indicates whether $x(p) \leq x(q)$).  

\begin{definition}
\label{def:configutration}
Fix an arbitrary ordering of the 2-sets of robots in $\R$. A \emph{configuration} of $\R$ is a sequence indicating, for each 2-set $\{R, R'\}$ of robots in $\R$, considered in the prescribed order, the relative order of $R$ with respect to $R'$. A \emph{realization} of a configuration $C$ is an embedding of the robots in $\R$ such that the relative order of any two robots in $\R$ conforms to that described by $C$ and the robots in the embedding are pairwise nonoverlapping in their interiors.\footnote{We can alternatively use the notion of order types instead of  configurations, but this would result in a much worse upper bound on the number of moves.}
\end{definition}

\begin{proposition}\label{p:aaconfiguration} 
\label{lem:realizations} For any two realizations $\rho, \rho'$ of a configuration $C$, there is a sequence of at most $2k$ valid moves within the bounding box $\Gamma$ that translate the robots from their positions in $\rho$ to their positions in $\rho'$.
\end{proposition}

\begin{proof}
The sequence of moves consists of two contiguous subsequences, the first in the horizontal direction and the second in the vertical direction.  For the first subsequence of moves, we partition the robots in $\R$ into two subsets: those whose positions in $\rho'$ have smaller $x$-coordinates than their positions in $\rho$, referred to as $\R_{left}$, and those whose positions in $\rho$ have larger or equal $x$-coordinates than those in $\rho$, referred to as $\R_{right}$. We sort the robots in $\R_{left}$ in a nondecreasing order of the $x$-coordinates of their left vertical segments, and those in $\R_{right}$ in a nonincreasing order of the $x$-coordinates of their right vertical segments. For each robot in $\R_{left}$, considered one by one in the sorted order, we translate it horizontally so that its left (or equivalently right) vertical segment is aligned with the vertical extension of its left (or equivalently right) vertical segment in $\rho'$. We do the same for the robots in $\R_{right}$. Since $\Gamma$ is rectangular, both $\rho$ and $\rho'$ are in $\Gamma$, and the translations are axis-aligned, it follows that all the aforementioned translations are within $\Gamma$. 

After each robot in $\R$ has reached its vertical extension in $\rho'$, we again partition $\R$ into two subsets, those whose top horizontal segments in $\rho'$ are above theirs in $\rho$, referred to as $\R_{above}$, and those whose top horizontal segments in $\rho'$ are below theirs in $\rho$, referred to as $\R_{below}$. We sort $\R_{above}$ in nonincreasing order of the their top horizontal segments (i.e., in decreasing order of their $y$-coordinates) and those in $\R_{below}$ in nondecreasing order of their bottom horizontal segments. We then route the robots in each partition in the sorted order to its horizontal line in $\rho'$, and hence to its final destination.

It is obvious that the above sequence of moves has length at most $2k$, and that at the end each robot is at its final destination in $\rho'$. Therefore, what is left is arguing that no move in the above sequence causes collision. It suffices to argue that for the subsequence of horizontal moves; the argument for the vertical moves is analogous. 

Let $R$ and $R'$ be two robots that are translated horizontally. Suppose first that both $R$ and $R'$ are moving in the same direction (either both in the direction of $\overrightarrow{x'x}$ or $\overrightarrow{xx'}$), say in the direction of $\overrightarrow{xx'}$. Assume, w.l.o.g., that $R$ is moved before $R'$, and hence, its left vertical segment has a smaller $x$-coordinate than that of $R'$. Since $\rho$ is a realization of $C$, and since we move the robots in the described sorted order, clearly $R$ cannot collide with $R'$ when $R$ was translated horizontally while $R'$ was stationary. Consider now the horizontal translation of $R'$. Since the relative position of $R$ and $R'$ are the same in $\rho$ and $\rho'$, and since the vertical segments of $R$ are now positioned on their vertical lines in $\rho'$, if $R'$ collides with $R$ during the translation of $R'$ then the relative position of the vertical segments of $R'$ with respect to those in $R$ (and hence, the relative positions of $R'$ w.r.t.~$R$) will have to change and will never change back. Since $R$ and $R'$ have the same relative order in $\rho$ and $\rho'$, this collision cannot occur. The argument is similar when $R$ and $R'$ are moving in opposite directions. This completes the proof.
\end{proof}

\begin{proposition}
\label{prop:boundedconfigurations}
Let $\I=(\R, \VVV, k, \ell, \Gamma)$ be a feasible instance of $\bigplus$-\srect{}-$\msquare$ or $\bigplus$-\prect{}-$\msquare$. Then there is a schedule for $\I$ of length at most $2k \cdot 5^{k(k-1)}$.  
\end{proposition}

\begin{proof}
Since the upper bound on the number of serial moves is also an upper bound on the number of parallel moves, it suffices to prove the theorem for $\bigplus$-\srect{}-$\msquare$.

Suppose that $\I$ is feasible and consider an optimal schedule for $\I$. Observe that there are five possible relative positions of any two axis-aligned rectangles along each of the two axes, and thus at most twenty five possible relative orderings for any two robots. It follows that the total number of (distinct) configurations is at most $25^{{k \choose 2}}$, where ${k \choose 2}$ is the number of 2-sets of $\R$. Since the schedule under consideration is optimal, the same configuration can appear at most $2k$ times in the schedule; otherwise, by Proposition~\ref{lem:realizations}, the schedule could be ``shortcut'' by considering a sequence of at most $2k$ moves that would translate the robots from the realization of the first occurrence of such a configuration to its last occurrence, thus reducing the length of schedule and contradicting its optimality. It follows that there is a schedule for $I$ of length at most $2k \cdot 
25^{{k \choose 2}}=2k \cdot 5^{k(k-1)}$. 
\end{proof}

We note that the obtained upper bounds on the length of the schedule of a feasible instance already suggest an approach for obtaining \FPT{} algorithms, namely that of enumerating all sequences of configurations and then checking their realizability. This approach results in much worse running times than the ones for the \FPT{} algorithms presented in this paper, and reveals less about the structure of the solution than the presented results.  

\fi

\section{Axis-Aligned Motion}\label{sec:grid}
In this section, we prove a structural result about $\bigplus$-\srect{}. This result, in particular, and the upper bound on the number of moves imply that $\bigplus$-\srect{} is \FPT{} parameterized by the number $k$ of robots.
In brief, the structural result states that, in order to obtain an optimal schedule to an instance of $\bigplus$-\srect, it is enough to restrict the robots to move along the lines of an axis-aligned grid (i.e., a collection of horizontal and vertical lines of the plane), that can be determined from the input instance. Moreover, the number of lines in the grid is a computable function of the number of robots, and the robots' moves will be defined using intersections of the grid lines. 
 
\begin{definition}
\label{def:grid}
Let $\I=(\R, \VVV, k, \ell)$ be an instance of $\bigplus$-\srect{}. We define an axis-aligned grid $G_{\I}$, associated with the instance $\I$, as follows. 

\begin{itemize} 
\item Initialize $G_{\I}$ to the set of horizontal and vertical lines through the starting and final positions of the centers of the robots in $\R$; call these lines the \textit{basic} grid lines. 
\item Add to $G_{\I}$ the lines which are defined using ``stackings'' of robots on the basic lines as follows; see Figure~\ref{fig:stacking}. Let $b \in G_{\I}$ be a vertical basic line with $x$-coordinate $x(b)$, and $w_b$ be the width of the robot whose center could be on $b$. For each number $1\leq i \leq \ell$, and each $i$-multiset $\{R_{j_1}, \ldots, R_{j_i} \}$ of robots, and for each choice of a horizontal width $w$ of a robot, add to $G_{\I}$ the two vertical lines with $x$-coordinates $x(b) \pm (w_b/2 + w/2 + \sum_{r=1}^{i} x(R_{j_r}))$. 
\item Add to $G_{\I}$ the analogous lines for the horizontal basic lines.  
\end{itemize}
\end{definition}

\begin{figure}
    \centering
    \includegraphics[scale=0.6]{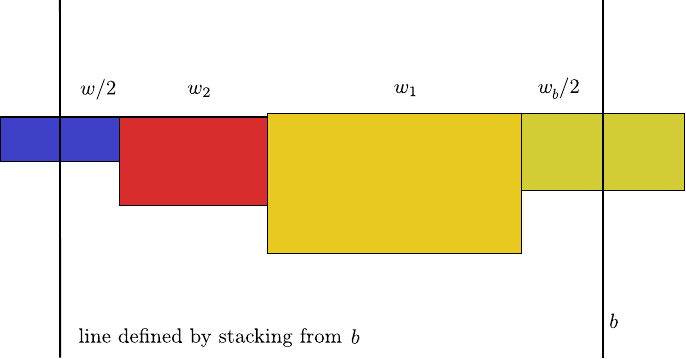}
    \caption{An illustration of a stacking to define new vertical lines.}
    \label{fig:stacking}
\end{figure}

\begin{theorem}\label{t:aastructure}
    Every instance $\I=(G, \R, k, \ell)$ of $\bigplus$-\srect{} has an optimal schedule in which every robot's move is between two grid points along a grid line in $G_{\I}$. The number of vertical (resp.~horizontal) lines in $G_{\I}$ is at most $k^3 \cdot 2^{k+\ell+1}$. 
\end{theorem}

\begin{proof}
We argue by induction on the number $\ell'$ of moves in the schedule. If $\ell'=1$, then the schedule has a single move that must be along a line defined by both the starting and ending positions of a robot in $\R$, and the statement is true in this case.
Thus, assume henceforth that the statement of the theorem is true when the optimal schedule has at most $\ell'-1$ moves. Let $R$ be the robot that
performs the first move (in the schedule) from some point $p_1$ to some point $p_2$, and assume, w.l.o.g., that the move is horizontal in the direction of $\overrightarrow{xx'}$. We define a new problem instance $\I'$, which is the same as $\I$, with the exceptions that in $\I'$ the robot $R$ now has starting position $p_2$ and the upper bound on the number of moves is $\ell'-1$.
 Let $G_{\I}$ and $G_{\I'}$ be the grids associated with instances $\I$ and $\I'$, respectively, as defined in Definition~\ref{def:grid}.
 The instance $\I'$ has a schedule of $\ell'-1$ moves and hence, there is a schedule for $\I'$ such that each robot moves along a grid line in $G_{\I'}$. 

The lines in the set $L_2  = G_{\I'}-G_{\I}$ are basic vertical grid lines defined by $R$ being at $p_2$ plus all the lines defined by stackings of these lines. Note that $L_2$ contains only vertical lines and that $G_{\I'}-L_2 \subseteq G_{\I}$. 
Let $G'$ be the set of grid lines in $G_{\I'}-L_2$ union the set of vertical lines obtained by stacking every robot on every line in $G_{\I'}-L_2$. Observe that $G' \subset G_{\I}$, and that we are allowed to perform this additional stacking operation since the construction of $G_\I$ involves $\ell'$ stackings, whereas the construction of  $G_\I'$ involves $\ell'-1$ stackings.

From among all schedules of length $\ell'-1$ for $\I'$ along the grid lines $G_{\I'}$, consider a schedule that uses the maximum number of grid lines from $G'$.

Note that the move of $R$ from $p_1$ to $p_2$ and the schedule for $\I'$ give an optimal schedule for the original problem instance $\I$; however, this schedule is not along the grid lines of $G_{\I}$, and some robots may move along the lines in $L_2$.

Let $M$ be the set of segments of the grid lines in $L_2$ traversed by robot moves that are along the lines of $L_2$; note that all of these are vertical.
Now we push back the robot $R$ from $p_2$ towards $p_1$. Let $R(p)$ denote the robot $R$ located at point $p \in p_1p_2$. The move from $p_1$ to $p$ remains valid, however, there might be intersections between $R(p)$ and other robots in future moves, and between this move and future positions of $R$ itself. 

When $p=p_2$, we have a schedule, and thus no collisions exist. From the construction of the grid lines, it follows that when the robot $R$ moves, the (updated) grid lines in $L_2$ move by exactly the same distance in the same direction. We now move all the segments of $M$ also in the same direction and distance, i.e., we move the grid lines in $L_2$ together with all the robot moves along them; see Figure~\ref{fig:illustration}. As we push back $R$ towards $p_1$, we stop the first time that the right edge of some robot, say $Q$, that travels along a segment in $M$, hits a vertical line that is defined by the left edge of a robot $Q'$ located on a line in $G_{\I'}-L_2$, and hence cannot be pushed further without potentially introducing a collision. Now the center of $Q$ is positioned at a line that is defined by some stacking of $G_{\I'}-L_2$, and hence is a vertical line of $G'$. Since we have not introduced any collisions during the pushing, we have obtained a grid schedule that uses more grid lines from $G'$, which contradicts the maximality of the chosen schedule for $\I'$. Therefore, there is a schedule in which all moves are along $G'$.

\begin{figure}

    \centering
    \includegraphics{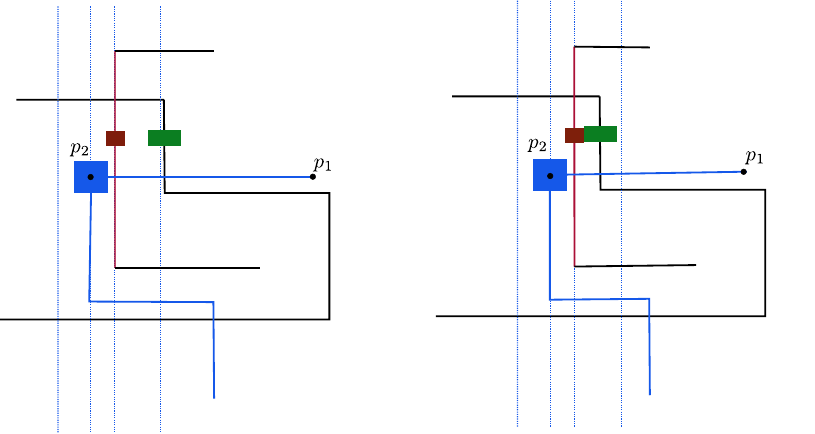}
    \caption{An illustration for the proof of Theorem~\ref{t:aastructure}. The blue broken line is the path traversed by the center of $R$ in a schedule. The black broken lines are the paths for the centers of other robots. The vertical dashed blue lines are grid lines in $L_2$. The red segment is a move along a line in $L_2$. The blue rectangle is $R$ after the first move, and the green and red rectangles are some other robots at some time during their movement. After moving the lines in $L_2$ to the right, the red segment becomes a stacking over lines not in $L_2$.}
    \label{fig:illustration}
\end{figure}

Finally, we  upper bound the number of lines in the grid. We upper bound the number of vertical lines; the upper bound on the number of horizontal lines is the same. Let ${\cal L}_V$ be the set of vertical lines in the grid, and 
${\cal L}_H$ be that of the horizontal lines.

The number of starting and ending positions of the robots is $2k$, and hence the number of basic vertical lines is at most $2k$. For each $i$, where $1 \leq i \leq \ell$, and for each basic vertical line $b$, we fix two robots: the one of width $w_b$ whose center could fall on $b$ and the one of width $w$ whose center could fall on the newly-defined vertical line based on the stacking. There are ${k \choose 2} \leq k^2/2$ choices for these two robots. Afterwards, we enumerate each selection of an $i$-multiset of robots, and for each $i$-multiset $\{R_{j_1}, \ldots, R_{j_i}\}$, we add the two vertical lines with offset $(w_l/2 + w/2 + \sum_{k=1}^{i} x(R_{j_k}))$ to the left and right of $b$. Note that this offset is determined by the two fixed robots and the $i$-multiset of robots, and hence by the two fixed robots and the set of robots in this multiset together with the multiplicity of each robot in this set. Therefore, in step $i$, to add the vertical lines, we can enumerate every 2-set of robots, each $r$-subset of robots, where $r \leq i$, order the robots in the $r$-set, and for each of the at most $2^k$ $r$-subsets of robots, enumerate all partitions $(\alpha_1, \ldots, \alpha_r)$ of $i$ into $r$ parts. The number of such partitions is at most $2^i$. Hence, the number of vertical lines we add in step $i$ is at most $2k \cdot k^2/2 \cdot 2^k \cdot 2^i$, and the total number of vertical lines we add over all the $\ell$ steps is at most $k^3\cdot 2^{k}\cdot \sum_{i=1}^{\ell} 2^i \leq k^2 \cdot 2^{k}\cdot   2^{\ell+1} =k^3 \cdot 2^{k+\ell+1}$. It follows that $|{\cal L}_V| \leq k^3 \cdot 2^{k+\ell+1}$ and 
$|{\cal L}_H| \leq k^3 \cdot 2^{k+\ell+1}$ as well.
Therefore, the total number of lines in the grid is $\Oh(k^3 \cdot 2^{k+ \ell})$. 
\end{proof}

\ifshort
\begin{theorem}[Appendix]
\label{thm:maingrid}
$\bigplus$-\srect{}, parameterized by the number of robots, can be solved in $\Oh^*(k^{16k} \cdot 2^{20k^2+8k})$ time, and hence is \FPT. 
\end{theorem}
\fi
 \iflong
\begin{theorem}
\label{thm:maingrid}
$\bigplus$-\srect{}, parameterized by the number of robots, can be solved in $\Oh^*(k^{16k} \cdot 2^{20k^2+8k})$ time, and hence is \FPT. 
\end{theorem}
\fi
\ifshort
\begin{proof}
By the above theorem and by Proposition~\ref{prop:free}, the number of horizontal/vertical lines in the grid is at most $k^{3} \cdot 2^{5k+1}$. The algorithm enumerates all potential schedules of length at most $\ell \leq 4k$ along the grid. The running time of the algorithm is $\Oh^*(k^{16k} \cdot 2^{20k^2+8k})$. 
\end{proof}
\fi

\iflong
\begin{proof}
By the above theorem and by Proposition~\ref{prop:free}, the total number of horizontal/vertical lines in the grid is $\Oh^*(k^{3} \cdot 2^{5k+1})$. The algorithm enumerates all potential schedules of length at most $\ell \leq 4k$ along the grid. To do so, the algorithm enumerates all sequences of at most $\ell \leq 4k$ moves of robots along the grid lines. To choose a move, we first choose a robot from among the $k$ robots, then choose the direction of the move (i.e., horizontal or vertical), and finally choose the grid line in either ${\cal L}_V$ (if the move is horizontal) or ${\cal L}_H$ (if the move is vertical) that the robot moves to. The number of choices per move is at most $k \cdot 2 \cdot k^3 2^{5k+1}$. It follows that the number of sequences that the algorithm needs to enumerate in order to find a desired schedule (if it exists) is at most $(k^4 \cdot 2 \cdot 2^{5k+1})^{\ell} \leq (k^4 \cdot 2^{5k+2})^{4k}$ by Proposition~\ref{prop:free}. Since the processing time over any given sequence is polynomial, it follows that the running time of the algorithm is $\Oh^*(k^{16k} \cdot 2^{20k^2+8k})$. 
\end{proof}
\fi
The following result is also a byproduct of our structural result, since one can, in polynomial time, ``guess'' and ``verify'' a schedule of length at most $\ell$ to an instance of $\bigplus$-\srect{} based on the grid corresponding to the instance:

\begin{corollary}
$\bigplus$-\srect{} is in \NP. 
\end{corollary}

The above corollary will be complemented with Theorem~\ref{thm:nphard} in Section~\ref{sec:nphardfixeddirections} to show that $\bigplus$-\srect{} is \NP-complete.

\ifshort
\section{An \FPT{} Algorithm When the Directions are Given}
\label{sec:general}
In this section, we give an \FPT{} algorithm for the case of axis-aligned rectangles that serially translate along a given (i.e., part of the input) fixed-cardinality set of directions. We first start by discussing the case where the robots move in the free plane, and then explain how the algorithm extends to the case where the robots are confined to a bounding box.

Let $\I=(R, \VVV, k, \ell)$ be an instance of \srect{}, where $\R=\{R_1, \ldots, R_k\}$ is a set of axis-aligned rectangular robots, and $\VVV=\{\overrightarrow{\theta_1}, \ldots, \overrightarrow{\theta_c}\}$, where $c \in \nat$ is a constant, is a set of unit vectors; we assumed herein that $c$ is a constant, but in fact, the results hold for any set of directions whose cardinality is a function of $k$.
Let $(s_{i}^{1}, s_{i}^{2})$ be the coordinates of the initial position of the center of $R_i$ and $(t_{i}^{1}, t_{i}^{2})$ be those of its final destination. 
We present a nondeterministic algorithm for the problem that makes a function of $k$ many guesses. The purpose of doing so is two fold. First, it serves the purpose of proving that \srect{} is in \NP{} since the nondeterministic algorithm runs in polynomial time (assuming that $|\VVV|$ is a constant or polynomial in $k$). Second, it will render the presentation of the algorithm much simpler. We will then show in Theorem~\ref{thm:lpfpt} how to make the algorithm deterministic by enumerating all possibilities for its nondeterministic guesses, and analyze its running time. 

The algorithm consists of three main steps: (1) guess the order in which the $k$ robots move in a schedule of length $\ell$ (if it exists); (2) guess the direction (i.e., the vector in $\VVV$) of each move; and (3) use Linear Programming (LP) to check the existence of corresponding amplitudes for the unit vectors associated with the $\ell$ moves that avoid collision.

We start by guessing the exact length, w.l.o.g.~call it $\ell$ (since it is a number between $0$ and $\ell$), of the  schedule sought. We then guess a sequence of $\ell$ events ${\cal E}=\langle e_1, \ldots, e_{\ell}\rangle$, where each event is a pair $(R_i, \overrightarrow{v_j})$, $i \in [k], j \in [c]$, that corresponds to a move/translation of a robot $R_i \in \R$ along a vector $\overrightarrow{v_j} \in \VVV$ in the sought schedule. 
The remaining part of the algorithm is to check if there is a schedule of length $\ell$ that is ``faithful'' to the guessed sequence ${\cal E}$
of events. 
That is, a schedule in which the robots' moves, and the translation in each move, correspond to those in ${\cal E}$.  
To do so, we will resort to LP. Basically, we will rely on LP to give us the exact translation vector (i.e., the amplitude) in each event $e_i$, $i \in [\ell]$, while ensuring no collision, in case a schedule of length $\ell$ exists. 

For each event $e_i=(R, \overrightarrow{v})$, $i \in [\ell]$, we introduce LP variables $x_{i}, y_{i}$ to encode the coordinates $(x_{i}, y_{i})$ of the center of $R$ at the beginning of the event. We also introduce an LP variable $\alpha_i > 0$ that encodes the amplitude of the translation of $R$ in the direction $\overrightarrow{v}$ in $e_i$. 

We form a set of LP instances such that the feasibility of one of them would produce the desired schedule, and hence, would imply a solution to instance $\I$. We explain next how this set of LP instances is formed. The LP constraints will stipulate the following conditions:

\begin{itemize}

\item[(i)] Each robot ends at its final destination.

\item[(ii)] Each robot starts at its initial position, and the starting position of robot $R_i$ in $e_q=(R_i, \overrightarrow{v}$) is the same as its final position after $e_p=(R_i, \overrightarrow{v'})$, where $e_p$ is the previous event to $e_q$ in ${\cal E}$ involving $R_i$ (i.e., $p$ is the largest index smaller than $q$). 


\item[(iii)] The translation in each $e_i$ is collision free.

\end{itemize}

 \begin{figure}[htbp]
 \centering
 
\begin{tikzpicture}

    \node at (1, 3)   (a){};
    \node[label=$a$] at (1, 2.9) {};
    \node at (1, 1)   (b){};
    \node[label=$b$] at (1, 0.4) {};
    \node at (4, 1)   (c) {};
    \node[label=$c$] at (4, 0.4) {};
    \node at (4, 3)   (d) {};
    \node[label=$d$] at (4, 2.9) {};

\node at (6, 6)   (a') {};
\node[label=$a'$] at (6, 5.9) {};

    \node at (6, 4)  (b'){};
    \node[label=$b'$] at (6,3.4) {};
    \node at (9, 4)  (c') {};
    \node[label=$c'$] at (9.2, 3.5){};
    \node at (9, 6)   (d') {};
    \node[label=$d'$] at (9.1, 5.85){};

       \node at (2.5, 2)   (o){$o$};
    \node at (7.5, 5)   (o'){$o'$};

                    \draw[line width =1mm] (1,3) -- (1,1);
                    \draw[line width =1mm] (1,1) -- (4,1);
  \draw[dashed, line width =0.5mm] (4,1) -- (4,3);
   \draw[dashed, line width =0.5mm] (1,3) -- (4,3);
   \draw[line width =1mm] (6,6) -- (9,6);
   \draw[line width =1mm] (9,4) -- (9,6);
    \draw[line width =1mm] (4,1) -- (9,4);
    \draw[line width =1mm] (1,3) -- (6,6);
   \draw[dashed, line width =0.5mm] (6,4) -- (9,4);
   \draw[dashed, line width =0.5mm] (b') -- (a');

     \draw[->, ultra thick, blue,  arrows={-latex}]  (o) -- (o');
 
\end{tikzpicture}
 
\caption{Illustration of the trace of a rectangle $abcd$ with respect to a vector $\overrightarrow{v}=\overrightarrow{oo'}$. Rectangle $a'b'c'd'=$  \translate($abcd, \overrightarrow{v})$ and the polygon $abcc'd'a'$, shown with solid lines, is \trace($abcd, \overrightarrow{v}$). Observe that the edges of a trace are either edges of the rectangles, or are parallel to $\overrightarrow{v}$.}
\label{fig:trace}
\end{figure}
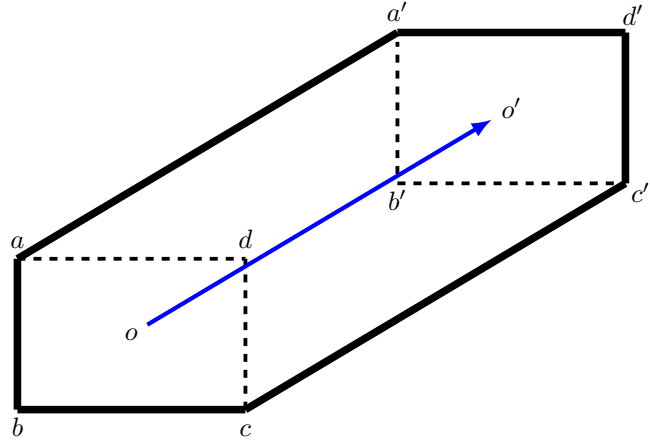

Conditions (i) and (ii) are easy to enforce using linear constraints (see Appendix).
We next discuss how the condition in $(iii)$ can be enforced. For a robot $R$ and a vector $\overrightarrow{v}$, denote by \emph{trace}($R, \overrightarrow{v}$) the boundary of the polygonal region of the plane covered during the translation of $R$ by the vector $\overrightarrow{v}$; see Figure~\ref{fig:trace} for illustration. It is clear that  \trace($R, \overrightarrow{v}$) is a polygon whose edges are either edges of $R$, or edges of  \translate($R, \overrightarrow{v})$, or line segments formed by a vertex of $R$ and a vertex of  \translate($R, \overrightarrow{v})$ whose slope is equal to that of $\overrightarrow{v}$. Therefore, if $R$ and $\overrightarrow{v}$ are fixed, then the slope of each edge of \trace($R, \overrightarrow{v}$) is fixed as well (i.e., independent of the LP variables). 

Now observe that robot $R_j \in \R$ does not collide with $R \in \R$ during the translation of $R$ by a vector $\alpha \cdot \overrightarrow{v}$, where $\overrightarrow{v} \in \VVV$, if and only if  no edge of $R_j$ and an edge of \trace($R, \alpha \cdot \overrightarrow{v}$) intersect in their interior. To stipulate that event $e_i=(R, \overrightarrow{v})$ is collision free, for each pair of edges $(pq, rs)$, where $pq$ is an edge of \trace($R, \alpha \cdot \overrightarrow{v}$) and $rs$ is an edge of $R_j$, we would like to add a linear constraint stipulating that the interiors of $rs$ and $pq$ do not intersect. If the slopes of the straight lines determined by $pq$ and $rs$ are the same, which we could check since the two slopes are given/fixed, then no such constraint is needed for this pair. Suppose now that the two straight lines $(rs)$ and $(pq)$ intersect at a point $\eta = (x_0, y_0)$; we  add linear constraints to stipulate that point $\eta$ does not lie in the interior of both segments $rs$ and $pq$, and hence the two segments do not intersect. To do so, we guess (i.e., branch into) one of the following four cases. Let $r=(x_r, y_r), s=(x_s, y_s), p=(x_p, y_p), q=(x_q, y_q)$ and assume, without loss of generality, that $x_p \leq x_q$ and that $x_r \leq x_s$.

\noindent {\bf Case (1):} Point $\eta$ is exterior to $pq$ and $x_0 \leq x_p$.

  \noindent {\bf Case (2):} Point $\eta$ is exterior to $pq$ and $x_0 \geq x_q$.

\noindent {\bf Case (3):} Point $\eta$ is exterior to $rs$ and $x_0 \leq x_r$.

\noindent {\bf Case (4):} Point $\eta$ is exterior to $rs$ and $x_0 \geq x_s$.

Note that $pq$ and $rs$ do not intersect in their interior if and only if (at least) one of the above cases holds. The algorithm guesses which case of the above four holds, and adds to the LP linear constraints stipulating the conditions of the guessed case. For instance, suppose that the algorithm guesses that Case (1) holds. Let $\beta, \gamma$ be the slopes of lines $(pq)$ and $(rs)$, respectively, and note that $\beta$ and $\gamma$ are known/fixed at this point . It is easy to verify that $x_0=(y_s-y_p +\beta x_p -\gamma x_s)/(\beta-\gamma)$. Therefore, to enforce the conditions in Case (1), we add to the LP the linear constraint:

\[ (y_s-y_p)+\beta x_p -\gamma x_s \leq (\beta-\gamma) x_p.\]

For each event $e_i=(R, \overrightarrow{v})$, and for each robot $R_j \in \R$, where $R_j \neq R$, and for each pair of edges $(pq, rs)$, where $pq$ is an edge of \trace($R, \alpha \cdot \overrightarrow{v}$) and $rs$ is an edge of $R_j$, the algorithm guesses which case of the above four cases applies and adds the corresponding linear constraint. The algorithm then solves the resulting LP. If the LP has a solution, then so does the instance $\I$. If the LP is not feasible, then the algorithm rejects the instance.

\begin{theorem}[Appendix]
\label{thm:lpfpt}
Given an instance $(\R, \VVV, k, \ell)$ of \srect{}, in time $\Oh^*((4^{32k} \cdot k \cdot |\VVV|)^{4k})$, we can construct a solution to the instance or decide that no solution exists, and hence \srect{} is \FPT.
\end{theorem}

\begin{corollary}\label{cor:NP}
     \srect is in \NP.
 \end{corollary}

\begin{proof}
The number of guesses made by the nondeterministic algorithm is polynomial.
\end{proof}

Next we discuss $\srect{}$-$\msquare$, in which the robots are confined to a bounding box. In this case, the problem becomes \PSPACE-hard as we observe in Section~\ref{sec:nphardfixeddirections}. It is easy to see that the LP part of the above approach can be easily modified to work for any rectangular bounding box by adding linear constraints stipulating that all rectangles resulting from the translations are confined to the box. (Basically, we only need to add constraints stipulating that the $x$/$y$-coordinate of each point are within the  vertical/horizontal lines of the bounding box.) The only issue is upper bounding the number of moves, $\ell$, in a feasible schedule. 

For the case of axis-aligned motion, that is, $\bigplus$-\srect{}-$\msquare$, Proposition~\ref{prop:boundedconfigurations} provides us with an upper bound of $2k \cdot 5^{k(k-1)}\leq 2k \cdot 5^{k^2}$ on $\ell$ in case the instance is feasible. Note that if the instance is not feasible, then the algorithm will end up rejecting the instance.  

\begin{theorem}
\label{thm:mainboxaxisaligned}
Given an instance $(\R, \VVV, k, \ell, \Gamma)$ of $\bigplus$-\srect{}-$\msquare$, in time $\Oh^*(5^{k^2}\cdot(4^{32k+1} \cdot k)^{2k \cdot 5^{k^2}})$, we can construct a solution to the instance or decide that no solution exists, and hence $\bigplus$-\srect{}-$\msquare$ is \FPT.
\end{theorem}

\fi

\iflong
\section{An \FPT{} Algorithm When the Directions are Given}
\label{sec:general}
In this section, we give an \FPT{} algorithm for the case of axis-aligned rectangles that serially translate along a given (i.e., part of the input) fixed-cardinality set of directions. We first start by discussing the case where the robots move in the free plane, and then explain how the algorithm extends to the case where the robots are confined to a bounding box.

Let $\I=(R, \VVV, k, \ell)$ be an instance of \srect{}, where $\R=\{R_1, \ldots, R_k\}$ is a set of axis-aligned rectangular robots, and $\VVV=\{\overrightarrow{\theta_1}, \ldots, \overrightarrow{\theta_c}\}$, where $c \in \nat$ is a constant, is a set of unit vectors; we assumed herein that $c$ is a constant, but in fact, the results hold for any set of directions whose cardinality is a function of $k$.
Let $(s_{i}^{1}, s_{i}^{2})$ be the coordinates of the initial position of the center of $R_i$ and $(t_{i}^{1}, t_{i}^{2})$ be those of its final destination. 
We present a nondeterministic algorithm for solving the problem that makes a function of $k$ many guesses. The purpose of doing so is two fold. First, it will serve the purpose of proving the membership of \srect{} in \NP{} since the nondeterministic algorithm runs in polynomial time (assuming that $|\VVV|$ is a constant or polynomial in $k$). Second, it will render the presentation of the algorithm much simpler. We will then show in Theorem~\ref{thm:lpfpt} how to make the algorithm deterministic by enumerating all possibilities for its nondeterministic guesses, and analyze its running time. 

At a high level, the algorithm consists of three main steps: (1) guess the order in which the $k$ robots move in a schedule of length $\ell$ (if it exists); (2) guess the direction (i.e., the vector in $\VVV$) of each move; and (3) use Linear Programming (LP) to check the existence of corresponding amplitudes for the unit vectors associated with the $\ell$ moves that avoid collision.

We start by guessing the exact length, w.l.o.g.~call it $\ell$ (since it is a number between $0$ and $\ell$), of the  schedule sought. We then guess a sequence of $\ell$ events ${\cal E}=\langle e_1, \ldots, e_{\ell}\rangle$, where each event is a pair $(R_i, \overrightarrow{v_j})$, $i \in [k], j \in [c]$, that corresponds to a move/translation of a robot $R_i \in \R$ along a vector $\overrightarrow{v_j} \in \VVV$ in the sought schedule. 

The remaining part of the algorithm is to check if there is a schedule of length $\ell$ that is ``faithful'' to the guessed sequence ${\cal E}$
of events. 
That is, a schedule in which the robots' moves, and the translation in each move, correspond to those in ${\cal E}$.

To do so, we will resort to LP. Basically, we will rely on LP to give us the exact translation vector (i.e., the amplitude) in each event $e_i$, $i \in [\ell]$, while ensuring no collision, in case a schedule of length $\ell$ exists. 

For each event $e_i=(R, \overrightarrow{v})$, $i \in [\ell]$, we introduce LP variables $x_{i}, y_{i}$ to encode the coordinates $(x_{i}, y_{i})$ of the center of $R$ at the beginning of event $e_i$. We also introduce an LP variable $\alpha_i > 0$ that encodes the amplitude of the translation of $R$ in the direction $\overrightarrow{v}$ for event $e_i$. 

We form a set of LP instances such that the feasibility of one of them would produce the desired schedule, and hence, would imply a solution to instance $\I$. We explain next how this set of LP instances is formed. The LP constraints will stipulate the following conditions:

\begin{itemize}

\item[(i)] Each robot ends at its final destination.

\item[(ii)] Each robot starts at its initial position, and the starting position of robot $R_i$ in $e_q=(R_i, \overrightarrow{v}$) is the same as its final position after $e_p=(R_i, \overrightarrow{v'})$, where $e_p$ is the previous event to $e_q$ in ${\cal E}$ involving $R_i$ (i.e., $p$ is the largest index smaller than $q$). 


\item[(iii)] The translation in each $e_i$ is collision free.

\end{itemize}

It is easy to see that the condition in (i) can be enforced by adding, for each robot $R_i \in \R$, two linear equality constraints, one stipulating that the total displacement of the $x$-coordinate of (the center of) $R_i$ with respect to all its translation moves, plus its initial $x$-coordinate, is equal to the $x$-coordinate of its final destination, and the other constraint stipulating the same with respect to the $y$-coordinate of the center of $R_i$. More specifically, suppose that the subsequence of events in ${\cal E}$ corresponding to the sequence of translations of $R_i$ is $\langle e_{p_1}, \ldots, e_{p_r}\rangle$; let $\overrightarrow{v_1}, \ldots, \overrightarrow{v_r} \in \VVV$ be the corresponding vectors of these events, and let $\alpha_1, \ldots, \alpha_r$ be the introduced LP variables corresponding to the amplitudes of these vectors, respectively. We add the following linear constraints:

\[ s_{i}^{1} + \sum_{j=1}^{r} \alpha_j  \cdot x(\overrightarrow{v_j}) = t_{i}^{1} \ \mbox{and} \  s_{i}^{2} + \sum_{j=1}^{r} \alpha_j  \cdot y(\overrightarrow{v_j}) = t_{i}^{2}.
\]

It is easy to see that linear equality constraints can be added to enforce the 
conditions in (ii). More specifically, we add linear constraints stipulating that, for each robot $R_i$, the coordinates of the center of $R_i$ in the first event in ${\cal E}$ that contains $R_i$ are $(s_{i}^{1}, s_{i}^{2})$. Moreover, for every two consecutive events $e_p, e_q \in {\cal E}$, $p < q$, involving a robot $R_i$, we add linear constraints stipulating that the coordinates of the position of $R_i$ after $e_p$ match those of its starting position before $e_q$ occurs.

 \begin{figure}[htbp]
 \centering
 
\begin{tikzpicture}

    \node at (1, 3)   (a){};
    \node[label=$a$] at (1, 2.9) {};
    \node at (1, 1)   (b){};
    \node[label=$b$] at (1, 0.4) {};
    \node at (4, 1)   (c) {};
    \node[label=$c$] at (4, 0.4) {};
    \node at (4, 3)   (d) {};
    \node[label=$d$] at (4, 2.9) {};

\node at (6, 6)   (a') {};
\node[label=$a'$] at (6, 5.9) {};

    \node at (6, 4)  (b'){};
    \node[label=$b'$] at (6,3.4) {};
    \node at (9, 4)  (c') {};
    \node[label=$c'$] at (9.2, 3.5){};
    \node at (9, 6)   (d') {};
    \node[label=$d'$] at (9.1, 5.85){};

       \node at (2.5, 2)   (o){$o$};
    \node at (7.5, 5)   (o'){$o'$};

                    \draw[line width =1mm] (1,3) -- (1,1);
                    \draw[line width =1mm] (1,1) -- (4,1);
  \draw[dashed, line width =0.5mm] (4,1) -- (4,3);
   \draw[dashed, line width =0.5mm] (1,3) -- (4,3);
   \draw[line width =1mm] (6,6) -- (9,6);
   \draw[line width =1mm] (9,4) -- (9,6);
    \draw[line width =1mm] (4,1) -- (9,4);
    \draw[line width =1mm] (1,3) -- (6,6);
   \draw[dashed, line width =0.5mm] (6,4) -- (9,4);
   \draw[dashed, line width =0.5mm] (b') -- (a');

     \draw[->, ultra thick, blue,  arrows={-latex}]  (o) -- (o');
 
\end{tikzpicture}
 
\caption{Illustration of the trace of a rectangle $abcd$ with respect to a vector $\overrightarrow{v}=\overrightarrow{oo'}$. Rectangle $a'b'c'd'=$  \translate($abcd, \overrightarrow{v})$ and the polygon $abcc'd'a'$, shown with solid lines, is \trace($abcd, \overrightarrow{v}$). Observe that the edges of a trace are either edges of the rectangles, or are parallel to $\overrightarrow{v}$.}
\label{fig:trace}
\end{figure}
We next discuss how the condition in $(iii)$ can be enforced. For a robot $R$ and a vector $\overrightarrow{v}$, denote by \emph{trace}($R, \overrightarrow{v}$) the boundary of the polygonal region of the plane covered during the translation of $R$ by the vector $\overrightarrow{v}$; see Figure~\ref{fig:trace} for illustration. More formally, if we let $R_x=$  \translate($R, x \cdot \overrightarrow{v})$, where $0 \leq x \leq 1$, then \trace($R, \overrightarrow{v}$) is the boundary of 
$\bigcup_{x=0}^{1} R_x$; that is, \trace($R, \overrightarrow{v}$)$= \partial (\bigcup_{x=0}^{1} R_x$). It is clear that  \trace($R, \overrightarrow{v}$) is a polygon whose edges are either edges of $R$, or edges of  \translate($R, \overrightarrow{v})$, or line segments formed by a vertex of $R$ and a vertex of  \translate($R, \overrightarrow{v})$ whose slope is equal to that of $\overrightarrow{v}$. Therefore, if $R$ and $\overrightarrow{v}$ are fixed, then the slope of each edge of \trace($R, \overrightarrow{v}$) is fixed as well (i.e., independent of the LP variables). 

Now observe that robot $R_j \in \R$ does not collide with $R \in \R$ during the translation of $R$ by a vector $\alpha \cdot \overrightarrow{v}$, where $\overrightarrow{v} \in \VVV$, if and only if  no edge of $R_j$ and an edge of \trace($R, \alpha \cdot \overrightarrow{v}$) intersect in their interior. To stipulate that event $e_i=(R, \overrightarrow{v})$ is collision free, for each pair of edges $(pq, rs)$, where $pq$ is an edge of \trace($R, \alpha \cdot \overrightarrow{v}$) and $rs$ is an edge of $R_j$, we would like to add a linear constraint stipulating that the interiors of $rs$ and $pq$ do not intersect. If the slopes of the straight lines determined by $pq$ and $rs$ are the same, which we could check since the two slopes are given/fixed, then no such constraint is needed for this pair. Suppose now that the two straight lines $(rs)$ and $(pq)$ intersect at a point $\eta = (x_0, y_0)$; we  add linear constraints to stipulate that point $\eta$ does not lie in the interior of both segments $rs$ and $pq$, and hence the two segments do not intersect. To do so, we guess (i.e., branch into) one of the following four cases. Let $r=(x_r, y_r), s=(x_s, y_s), p=(x_p, y_p), q=(x_q, y_q)$ and assume, without loss of generality, that $x_p \leq x_q$ and that $x_r \leq x_s$.

\noindent {\bf Case (1):} Point $\eta$ is exterior to $pq$ and $x_0 \leq x_p$.

  \noindent {\bf Case (2):} Point $\eta$ is exterior to $pq$ and $x_0 \geq x_q$.

\noindent {\bf Case (3):} Point $\eta$ is exterior to $rs$ and $x_0 \leq x_r$.

\noindent {\bf Case (4):} Point $\eta$ is exterior to $rs$ and $x_0 \geq x_s$.

Note that $pq$ and $rs$ do not intersect in their interior if and only if (at least) one of the above cases holds. The algorithm guesses which case of the above four holds, and adds to the LP linear constraints stipulating the conditions of the guessed case. For instance, suppose that the algorithm guesses that Case (1) holds. Let $\beta, \gamma$ be the slopes of lines $(pq)$ and $(rs)$, respectively, and note that $\beta$ and $\gamma$ are known/fixed at this point . It is easy to verify that $x_0=(y_s-y_p +\beta x_p -\gamma x_s)/(\beta-\gamma)$. Therefore, to enforce the conditions in Case (1), we add to the LP the linear constraint:

\[ (y_s-y_p)+\beta x_p -\gamma x_s \leq (\beta-\gamma) x_p.\]

For each event $e_i=(R, \overrightarrow{v})$, and for each robot $R_j \in \R$, where $R_j \neq R$, and for each pair of edges $(pq, rs)$, where $pq$ is an edge of \trace($R, \alpha \cdot \overrightarrow{v}$) and $rs$ is an edge of $R_j$, the algorithm guesses which case of the above four cases applies and adds the corresponding linear constraint. The algorithm then solves the resulting LP. If the LP has a solution, then the instance $\I$ has a solution. If the LP is not feasible, then the algorithm rejects the instance.

We have the following theorem:

\begin{theorem}
\label{thm:lpfpt}
Given an instance $(\R, \VVV, k, \ell)$ of \srect{}, in time $\Oh^*((4^{32k} \cdot k \cdot |\VVV|)^{4k})$, we can construct a solution to the instance or decide that no solution exists, and hence \srect{} is \FPT.
\end{theorem}

\begin{proof}
It is straightforward to verify that if there is a solution to the problem instance then there exists a guess corresponding to a feasible LP, and hence the algorithm is a nondeterministic algorithm that solves the problem correctly. We next analyze the time needed to simulate the algorithm deterministically.

First, noting that $\ell \leq 4k$ by Proposition~\ref{prop:free}, the algorithm branches into $\Oh(k)$ branches, each corresponding to a possible value 
$0 \leq \ell \leq 4k$, in order to find the correct number of moves/events in a schedule (if it exists); this branch adds only an $\Oh(k)$ multiplicative factor to the running time.

The number of sequences of $\ell$ events that we need to enumerate is at most 
$(k \cdot |\VVV|)^{\ell} \leq (k \cdot |\VVV|)^{4k}$. The algorithm will branch into each of these sequences. Fix such a sequence ${\cal E}$. For each event $(R, \overrightarrow{v})$ in this sequence, observe that the number of edges in \trace($R, \alpha \cdot \overrightarrow{v}$) is at most $8$. Each $R_j \neq R$ has four edges, and hence at most 32 pairs of line segments are considered for stipulating non-collision between $R$ and $R_j$, for a total of at most $32k$ pairs of segments over all robots in $\R$. For each pair of segments, the algorithm guesses which of the four possible cases is correct. Therefore, the algorithm branches into a total of $4^{32k}$ branches for each event, for a total of at most $(4^{32k})^{\ell} \leq (4^{32k})^{4k}$ branches over the sequence of events ${\cal E}$.  Once the LP is formed, it can be solved in polynomial time~\cite{lppolynomial1, lppolynomial}. It follows that the algorithm can be simulated by creating a search tree of size $\Oh^*((4^{32k} \cdot k \cdot |\VVV|)^{4k})$. Since at every node of the search tree all the operations that the algorithm performs can be implemented in polynomial time, it follows that the  nondeterministic algorithm can be simulated by a deterministic algorithm in time $\Oh^*((4^{32k} \cdot k \cdot |\VVV|)^{4k})$. 
\end{proof}

\begin{corollary}\label{cor:NP}
     \srect is in \NP.
 \end{corollary}

\begin{proof}
It is easy to verify that the number of guesses made by the nondeterministic algorithm is polynomial.
\end{proof}

Next we discuss $\srect{}$-$\msquare$, in which the robots are confined to a bounding box. In this case, the problem becomes \PSPACE-hard as we note in Section~\ref{sec:nphardfixeddirections}. It is easy to see that the LP part of the above approach can be easily modified to work for any rectangular bounding box by adding linear constraints stipulating that all rectangles resulting from the translations are confined to the box. (Basically, we only need to add constraints stipulating that the $x$/$y$-coordinate of each point are within the  vertical/horizontal lines of the bounding box.) The only issue is upper bounding the number of moves, $\ell$, in a feasible schedule by a function of $k$.

For the case of axis-aligned motion, that is, $\bigplus$-\srect{}-$\msquare$, Proposition~\ref{prop:boundedconfigurations} provides us with an upper bound of $2k \cdot 5^{k(k-1)}\leq 2k \cdot 5^{k^2}$ on $\ell$ in case the instance is feasible. Note that if the instance is not feasible, then the algorithm will end up rejecting the instance. Hence, by a similar run-time analysis to that in Theorem~\ref{thm:lpfpt}, we have:

\begin{theorem}
\label{thm:mainboxaxisaligned}
Given an instance $(\R, \VVV, k, \ell, \Gamma)$ of $\bigplus$-\srect{}-$\msquare$, in time $\Oh^*(5^{k^2}\cdot(4^{32k+1} \cdot k)^{2k \cdot 5^{k^2}})$, we can construct a solution to the instance or decide that no solution exists, and hence $\bigplus$-\srect{}-$\msquare$ is \FPT.
\end{theorem}
\fi

\ifshort
\section{An \FPT{} Algorithm for $\bigplus$-\ppcmp}
\label{sec:parallel}
In this section, using a similar approach to that in the previous section, we obtain an \FPT{} algorithm for $\bigplus$-\ppcmp{}.  The only major challenge now is to stipulate non-collision in the case of parallel motion. We again first discuss $\bigplus$-\ppcmp{} and then extend the \FPT{} result to $\bigplus$-\ppcmp{}-$\msquare$.

 Let $(\R=\{R_1, \ldots, R_k\}, \VVV, k, \ell)$ be an instance of $\bigplus$-\prect{}, where $\VVV$ is the set of unit vectors of the negative and positive $x$-axis and $y$-axis.  
We again present a nondeterministic algorithm for the problem, which will imply its membership in \NP. The algorithm proceeds in a similar fashion to that of \srect{} by guessing the exact number $\ell$ of moves in the sought schedule and then guessing the sequence 
 ${\cal E}=\langle e_1, \ldots, e_{\ell}\rangle$ of $\ell$ events corresponding to the schedule, with the exception that now each event --instead of containing the single robot that moves in that event and the direction of its translation -- contains a subset of robots and their corresponding directions of translations in $\VVV$; that is, each event $e_i$ is now a pair of the form $(S_i, V_i)$, $i \in [\ell]$, where $S_i \subseteq \R$ and $V_i \subseteq \VVV$ (and the implicit association between each robot in $S_i$ and its direction in $V_i$).

We again introduce, for each event $e_i=(S_i, V_i)$, and for each robot $R_j \in S_i$, variables $x_{j}, y_{j}$ to encode the coordinates $(x_{j}, y_{j})$ of the center of $R_j$ at the beginning of event $e_i$. We also introduce a variable $\alpha_j > 0$ that encodes the amplitude of the translation in the direction $\overrightarrow{v_j} \in V_i$ corresponding to the translation of $R_j$.   

Consider a pair $R_p, R_q \in S_i$, and let $\overrightarrow{v_p}, \overrightarrow{v_q}$ be the directions of their translations in $V_i$, respectively, and let $\alpha_p, \alpha_q$ be the LP variables corresponding to the amplitudes of $\overrightarrow{v_p}, \overrightarrow{v_q}$, respectively. We only discuss here how to encode in the LP that the translations of $R_p$ and $R_q$ in event $e_i$ are collision-free and refer to the full version of the paper for the complete details. The algorithm makes several guesses to distinguish all possible cases. We focus here only on the two most involved cases arising when \trace($R_p, \alpha_p \cdot \overrightarrow{v_p}$) and \trace($R_q, \alpha_q \cdot \overrightarrow{v_q}$) intersect in their interior. We make guesses to distinguish the two (out of several) subcases below, and the conditions under which they apply; refer to Figure~\ref{fig:brachingillustration}.

\noindent {\bf First Subcase:} $R_p$ is to the right of $R_q$ when $R_q$ reaches the horizontal line determined by the bottom-horizontal segment of $R_p$. We add the constraint:

\[(y_p - y(R_p)/2) - (y_q + y(R_q)/2)  \geq x_q + x(R_q)/2 - (x_p-x(R_p)/2).\]

\noindent {\bf Second Subcase:} $R_q$ is above $R_p$ when $R_p$ reaches the vertical line determined by the  left-vertical segment of $R_q$. We add the linear constraint:

\[(x_q - x(R_q)/2) - (x_p + x(R_p)/2)  \geq y_p + y(R_p)/2 - (y_q-y(R_q)/2).\]

\begin{figure}[h]
 
  \resizebox{0.35\textwidth}{!}{
\begin{tikzpicture}
       
    \draw[line width=1mm] (0, 0) rectangle (4, 2);
     \draw[line width=1mm] (1, -8) rectangle (3, -6);

     \node[label=$R_p$] at (-0.5, 0.6) {};
      \node[label=$R_q$] at (0.5, -7.4) {};

 \draw[<->, line width=0.5mm, blue]    (3,-5.8) -- (3,-0.2);
 \node[text=blue] at (2.6, -3) {$\Delta_V$};

 \draw[<->, line width=0.5mm, red]    (0,-0.3) --(2.8,-0.3);
 \node[text=red] at (1.5, -0.7) {$\Delta_H$};

 \draw[->, line width=0.5mm]    (-0.5,3) --(6,3);

 \draw[->, line width=0.5mm, black]    (-1,-8) --(-1,5);

     \node at (-2, 1) {$\alpha_q \cdot \overrightarrow{v_q}$};
     \node at (2, 3.5) {$\alpha_p \cdot \overrightarrow{v_p}$};
\end{tikzpicture}
 }
\resizebox{0.65\textwidth}{!}{
      \hspace*{0.1\linewidth}
  \begin{tikzpicture}
     
    \draw[line width=1mm] (0, 0) rectangle (4, 2);
     \draw[line width=1mm] (11, -3) rectangle (13, -1);

     \node[label=$R_p$] at (-0.5, 0.6) {};
      \node[label=$R_q$] at (10.5, -2) {};

 \draw[<->, line width=0.5mm, blue]    (13.5,-3) -- (13.5,2);
 \node[text=blue] at (13, 0) {$\Delta_V$};

 \draw[<->, line width=0.5mm, red]    (4,0) --(11,0);
 \node[text=red] at (7, 1) {$\Delta_H$};

 \draw[->, line width=0.5mm]    (-0.5,3) --(11,3);

 \draw[->, line width=0.5mm, black]    (14,-6) --(14,7);

     \node at (13.2, 3) {$\alpha_q \cdot \overrightarrow{v_q}$};
     \node at (6, 4) {$\alpha_p \cdot \overrightarrow{v_p}$};
 \end{tikzpicture}
 }
  
\caption{Illustration of the two subcases. In the left figure (first subcase), the distance $\Delta_V$ between the top edge of $R_q$ and the bottom edge of $R_p$ is larger than the distance $\Delta_H$ between the left edge of $R_p$ and the right edge of $R_q$. Hence, $R_p$ manages to ``escape'' $R_q$ in time. Similarly, in the right figure (second subcase), $R_q$ manages to escape $R_p$ in time.}
\label{fig:brachingillustration}
\end{figure}

\begin{theorem}[Appendix]
\label{thm:lpparallelfpt}
Given an instance $(\R, \VVV, k, \ell)$ of $\bigplus$-\prect{}, in time $\Oh^*(5^{2k^3}\cdot 8^{4k^2})$, we can compute a solution to the instance or decide that no solution exists, and hence $\bigplus$-\prect{} is \FPT.
\end{theorem}
 
\begin{corollary}
     $\bigplus$-\prect{} is in \NP.
 \end{corollary}
We extend the above result to $\bigplus$-\prect{}-$\msquare$, which is \PSPACE-hard (see Section~\ref{sec:nphardfixeddirections}): 
\begin{theorem}
\label{thm:lpparallelfptbox}
Given an instance $(\R, \VVV, k, \ell)$ of $\bigplus$-\prect{}-$\msquare$, in time $\Oh^*(5^{k^2}\cdot 8^{2k^2\cdot5^{k^2}} \cdot 5^{k^3\cdot 5^{k^2}})$, we can compute a solution to the instance or determine that no solution exists, and hence $\bigplus$-\prect{}-$\msquare$ is \FPT.
\end{theorem}
\fi

\iflong
\section{An \FPT{} Algorithm for $\bigplus$-\ppcmp}
\label{sec:parallel}
In this section, using a similar approach to that in the previous section, we obtain an \FPT{} algorithm for $\bigplus$-\ppcmp{}.  The only major difference and challenge is how to stipulate non-collision in the case of parallel motion. We again first discuss $\bigplus$-\ppcmp{} and then extend the \FPT{} result to $\bigplus$-\ppcmp{}-$\msquare$.

Let $\VVV=\{\overrightarrow{H}^{-}, \overrightarrow{H}^{+}, \overrightarrow{V}^{-}, \overrightarrow{V}^{+}\}$ denote the set of unit vectors of the negative $x$-axis, positive $x$-axis, negative $y$-axis, and positive $y$-axis, respectively, and let $(\R=\{R_1, \ldots, R_k\}, \VVV, k, \ell)$ be an instance of $\bigplus$-\prect{}.  

We again present a nondeterministic algorithm for solving the problem. The main purpose behind that is to conclude the membership of the problem in \NP. The algorithm proceeds in a similar fashion to that of \srect{} by guessing the exact number $\ell$ of moves in the sought schedule and then guessing the sequence 
 ${\cal E}=\langle e_1, \ldots, e_{\ell}\rangle$ of $\ell$ events corresponding to the schedule, with the exception that now each event---instead of containing the single robot that moves in that event and the direction of its translation---contains a subset of robots and their corresponding directions of translations in $\VVV$; that is, each event $e_i$ is now a pair of the form $(S_i, V_i)$, $i \in [\ell]$, where $S_i \subseteq \R$ and $V_i \subseteq \VVV$ (and the implicit association between each robot in $S_i$ and its direction in $V_i$).

We again introduce, for each event $e_i=(S_i, V_i)$, and for each robot $R_j \in S_i$, variables $x_{j}, y_{j}$ to encode the coordinates $(x_{j}, y_{j})$ of the center of $R_j$ at the beginning of event $e_i$. We also introduce a variable $\alpha_j > 0$ that encodes the amplitude of the translation in the direction $\overrightarrow{v_j} \in V_i$ corresponding to the translation of $R_j$. 

The linear constraints in the LPs stipulating that each robot ends at its final destination and that the translation of a robot $R_j \in S_i$ does not collide with any (stationary) robot $R \notin S_i$ remain the same. The only additional constraints that we need to impose are those stipulating that no collision takes place between any two robots in $S_i$, that is, between two robots that are moving simultaneously during event $e_i$. 

Consider a pair $R_p, R_q \in S_i$, and let $\overrightarrow{v_p}, \overrightarrow{v_q}$ be the directions of their translations in $V_i$, respectively, and let $\alpha_p, \alpha_q$ be the LP variables corresponding to the amplitudes of $\overrightarrow{v_p}, \overrightarrow{v_q}$, respectively. To encode in the LP that the translations of $R_p$ and $R_q$ in event $e_i$ are collision-free, we distinguish the following cases:

\noindent{\bf Case (A):} The translations of both $R_p$ and $R_q$ are horizontal. Here, we further distinguish the two subcases: (A)-1 $\overrightarrow{v_{p}}=\overrightarrow{v_{q}}$ and (A)-2 $\overrightarrow{v_{p}}=-\overrightarrow{v_{q}}$.

\noindent {\bf Subcase (A)-1:} Without loss of generality, we assume that $\overrightarrow{v_p}=\overrightarrow{v_{q}}=\overrightarrow{H}^{+}$; the other case where $\overrightarrow{v_p}=\overrightarrow{v_{q}}=\overrightarrow{H}^{-}$ is analogous. The algorithm guesses one of the following four cases (i.e., we branch into the following four cases). In the first case, $R_p$ is above $R_q$, that is, the bottom-horizontal segment of $R_p$ is not below the top-horizontal segment of $R_q$. In this case $R_p$ and $R_q$ do not collide during $e_i$, and we add to the LP the linear constraint stipulating the assumption of this case, which is:

\[y_p - y(R_p)/2 \geq y_q + y(R_q)/2.\]

The second case is that $R_p$ is below $R_q$, that is, the top-horizontal segment of $R_p$ is not above the bottom-horizontal segment of $R_q$, and again, this case leads to no collision. We add the constraint to the LP stipulating the assumption of this case, which is:  

\[y_p + y(R_p)/2 \leq y_q - y(R_q)/2.\]

The third case is that $R_p$ is to the left of $R_q$, that is, the right-vertical segment of $R_p$ is to the left of the left-vertical segment of $R_q$. This case may lead to collision. We add the constraint to the LP stipulating the assumption of this case, which is:  

\[x_p + x(R_p)/2 \leq x_q - x(R_q)/2.\]

We also add the constraint stipulating that no collision occurs, which is:

\[x_p + x(R_p)/2 + \alpha_p \leq x_q - x(R_q)/2 + \alpha_q.\]

The fourth case is that $R_p$ is to the right of $R_q$, that is, the right-vertical segment of $R_q$ is not to the right of the left-vertical segment of $R_p$. This case may lead to collision. We add the constraint to the LP stipulating the assumption of this case, which is:  

\[x_q + x(R_q)/2 \leq x_p - x(R_p)/2.\]

We also add the constraint stipulating that no collision occurs, which is:

\[x_q + x(R_q)/2 + \alpha_q \leq x_p - x(R_p)/2 + \alpha_p.\]

\noindent {\bf Subcase (A)-2:} Assume that $\overrightarrow{v_p}=\overrightarrow{H}^{+}$ and $\overrightarrow{v_{q}}=\overrightarrow{H}^{-}$; the treatment of the other case is analogous. As in subcase (A)-1, we guess one of four cases. The first two cases, in which we guess whether $R_p$ is above/below $R_q$, are treated exactly the same as in Case (A)-1. Hence, we only consider the third and fourth cases. 

The case where $R_p$ is to the right of $R_q$ leads to no collision, and we only add the constraint stipulating the assumption in this case, which is:

\[x_p - x(R_p)/2 \geq x_q + x(R_q)/2.\]

In the fourth case, $R_p$ is to the left of $R_q$ and this case may lead to collision. We add the constraint stipulating the assumption of this case, which is:

\[x_p + x(R_p)/2 \leq x_q - x(R_q)/2.\]

We also add the constraint stipulating that no collision occurs, which is:

\[x_p + x(R_p)/2 + \alpha_p \leq x_q - x(R_q)/2 - \alpha_q.\]

\noindent {\bf Case (B):} The translations of both $R_p$ and $R_q$ are vertical. The treatment of this case is very similar to that of Case (A), and hence is omitted. 

\noindent {\bf Case (C):} The translation of $R_p$ is horizontal and the translation of $R_q$ is vertical or vice versa. We only consider the case where $R_p$ translates horizontally and $R_q$ vertically; the other case is symmetrical. Assume that $\overrightarrow{v_p}=\overrightarrow{H}^{+}$ and $\overrightarrow{v_{q}}=\overrightarrow{V}^{+}$; the treatment of the other cases are analogous.

We guess one of the two subcases: (C)-1 \trace($R_p, \alpha_p \cdot \overrightarrow{v_p}$) and \trace($R_q, \alpha_q \cdot \overrightarrow{v_q}$) do not intersect in their interior and (C)-2 \trace($R_p, \alpha_p \cdot \overrightarrow{v_p}$) and \trace($R_q, \alpha_q \cdot \overrightarrow{v_q}$) intersect in their interior.

Subcase (C)-1 could happen in four possible scenarios; we guess the scenario and, in each case, add the linear constraints stipulating the assumption of that scenario.

\noindent {\bf Scenario (C)-1-1:} $R_q$ is above $R_p$. This is the same as the cases treated before: We add the constraint

\[y_q - y(R_q)/2 \geq y_p + y(R_p)/2.\]

\noindent {\bf Scenario (C)-1-2:} $R_q$ is to the left of $R_p$. This is also similar to one of the cases treated before: We add the constraint

\[x_q + x(R_q)/2 \leq x_p - x(R_p)/2.\]

\noindent {\bf Scenario (C)-1-3:} The translation of $R_q$ by vector $\alpha_q \cdot \overrightarrow{V}^{+}$ is below $R_p$. We add the constraint

\[y_q + y(R_q)/2 + \alpha_q \leq y_p - y(R_p)/2.\]

\noindent {\bf Scenario (C)-1-4:} The translation of $R_p$ by vector $\alpha_p \cdot \overrightarrow{H}^{+}$ is to the left of $R_q$. We add the constraint

\[x_p + x(R_p)/2 + \alpha_p \leq x_q - x(R_q)/2.\]

In Case (C)-2, we know that \trace($R_p, \alpha_p \cdot \overrightarrow{v_p}$) and \trace($R_q, \alpha_q \cdot \overrightarrow{v_q}$) intersect in their interior. We guess which of the two subcases below applies; see Figure~\ref{fig:brachingillustration}.

\noindent {\bf Subcase (C)-(2)-1:} $R_p$ is to the right of $R_q$ when $R_q$ reaches the horizontal line determined by the bottom-horizontal segment of $R_p$. We add the constraint:

\[(y_p - y(R_p)/2) - (y_q + y(R_q)/2)  \geq x_q + x(R_q)/2 - (x_p-x(R_p)/2).\]
 
\noindent {\bf Subcase (C)-(2)-2:} $R_q$ is above $R_p$ when $R_p$ reaches the vertical line determined by the    left-vertical segment of $R_q$. We add the linear constraint:

\[(x_q - x(R_q)/2) - (x_p + x(R_p)/2)  \geq y_p + y(R_p)/2 - (y_q-y(R_q)/2).\]

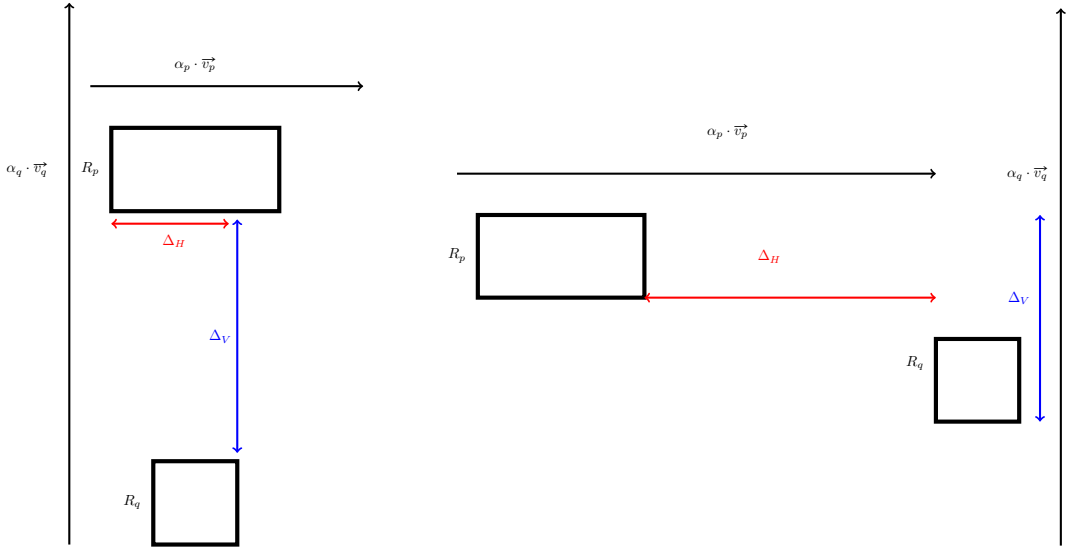
\begin{figure}[h]
 
  \resizebox{0.35\textwidth}{!}{
\begin{tikzpicture}
       
    \draw[line width=1mm] (0, 0) rectangle (4, 2);
     \draw[line width=1mm] (1, -8) rectangle (3, -6);

     \node[label=$R_p$] at (-0.5, 0.6) {};
      \node[label=$R_q$] at (0.5, -7.4) {};

 \draw[<->, line width=0.5mm, blue]    (3,-5.8) -- (3,-0.2);
 \node[text=blue] at (2.6, -3) {$\Delta_V$};

 \draw[<->, line width=0.5mm, red]    (0,-0.3) --(2.8,-0.3);
 \node[text=red] at (1.5, -0.7) {$\Delta_H$};

 \draw[->, line width=0.5mm]    (-0.5,3) --(6,3);

 \draw[->, line width=0.5mm, black]    (-1,-8) --(-1,5);

     \node at (-2, 1) {$\alpha_q \cdot \overrightarrow{v_q}$};
     \node at (2, 3.5) {$\alpha_p \cdot \overrightarrow{v_p}$};
\end{tikzpicture}
 }
\resizebox{0.65\textwidth}{!}{
      \hspace*{0.1\linewidth}
  \begin{tikzpicture}
     
    \draw[line width=1mm] (0, 0) rectangle (4, 2);
     \draw[line width=1mm] (11, -3) rectangle (13, -1);

     \node[label=$R_p$] at (-0.5, 0.6) {};
      \node[label=$R_q$] at (10.5, -2) {};

 \draw[<->, line width=0.5mm, blue]    (13.5,-3) -- (13.5,2);
 \node[text=blue] at (13, 0) {$\Delta_V$};

 \draw[<->, line width=0.5mm, red]    (4,0) --(11,0);
 \node[text=red] at (7, 1) {$\Delta_H$};

 \draw[->, line width=0.5mm]    (-0.5,3) --(11,3);

 \draw[->, line width=0.5mm, black]    (14,-6) --(14,7);

     \node at (13.2, 3) {$\alpha_q \cdot \overrightarrow{v_q}$};
     \node at (6, 4) {$\alpha_p \cdot \overrightarrow{v_p}$};
 \end{tikzpicture}
 }
  
\caption{Illustration of Subcases (C)-(2)-1 (left figure) and (C)-(2)-2 (right figure). In Subcase (C)-(2)-1, the distance $\Delta_V$ between the top edge of $R_q$ and the bottom edge of $R_p$ is larger than the distance $\Delta_H$ between the left edge of $R_p$ and the right edge of $R_q$. Hence, $R_p$ manages to ``escape'' $R_q$ in time. Similarly, in Subcase (C)-(2)-2, $R_q$ manages to escape $R_p$ in time. }
\label{fig:brachingillustration}
\end{figure}

We conclude with the following theorem:

\begin{theorem}
\label{thm:lpparallelfpt}
Given an instance $(\R, \VVV, k, \ell)$ of $\bigplus$-\prect{}, in time $\Oh^*(5^{2k^3}\cdot 8^{4k^2})$, we can compute a solution to the instance or decide that no solution exists, and hence $\bigplus$-\prect{} is \FPT.
\end{theorem}

\begin{proof}
It is is easy to verify the correctness of the algorithm. We only analyze the running time needed in order to simulate the nondeterministic algorithm.  

Again, branching into $\Oh(k)$ branches, each corresponding to a possible value 
$0 \leq \ell \leq 4k$ (see Proposition~\ref{prop:free}) to find the correct number of moves/events in the desired schedule adds only an $\Oh(k)$ multiplicative factor to the running time.

The number of sequences of $\ell$ events that we need to enumerate is at most 
$(2^k \cdot 4^k)^{\ell} \leq (8)^{4k^2}$. The algorithm will branch into each of these sequences. Fix such a sequence ${\cal E}$. For each event $(S_i, V_{i})$ in this sequence, there are ${k \choose 2} \leq k^2/2$ pairs of robots in $S_i$ to consider for collision. It is easy to see that the worst case corresponds to when $S_i = \R$. In this case, for each pair, the algorithm needs to branch into at most five cases (corresponding to branching into one of the two cases (C)-1, (C)-2, and in the (C)-1 case, further branching into four subcases).  Therefore, the algorithm branches into a total of $5^{k^2/2}$ branches for each event, for a total of at most $(5^{k^2/2})^{\ell} \leq (5^{2k^3})$ branches over the sequence of events ${\cal E}$. Once the LP is formed, it can be solved in polynomial time~\cite{lppolynomial1, lppolynomial}. It follows that the algorithm can be simulated by creating a search tree of size $\Oh^*(5^{2k^3} \cdot 8^{4k^2})$. Since at any node of the search tree all operations that the algorithm performs can be implemented in polynomial time, it follows that the  nondeterministic algorithm can be simulated by a deterministic algorithm in time $\Oh^*(5^{2k^3}\cdot 8^{4k^2})$. 
\end{proof}

\begin{corollary}
     $\bigplus$-\prect{} is in \NP.
 \end{corollary}
To extend the above result to $\bigplus$-\prect{}-$\msquare$, which is \PSPACE-hard (see Section~\ref{sec:nphardfixeddirections}), note that the LP part works seamlessly by adding linear constraints confining the moves to the bounding box.  We just need to modify the upper bound on the length $\ell$ of a schedule for a feasible instance by $2k \cdot 5^{k(k-1)}\leq 2k \cdot 5^{k^2}$ (see Proposition~\ref{prop:boundedconfigurations}). By modifying the run-time analysis in the proof of Theorem~\ref{thm:lpparallelfpt} accordingly, we obtain :

\begin{theorem}
\label{thm:lpparallelfptbox}
Given an instance $(\R, \VVV, k, \ell)$ of $\bigplus$-\prect{}-$\msquare$, in time $\Oh^*(5^{k^2}\cdot 8^{2k^2\cdot5^{k^2}} \cdot 5^{k^3\cdot 5^{k^2}})$, we can compute a solution to the instance or determine that no solution exists, and hence $\bigplus$-\prect{}-$\msquare$ is \FPT.
\end{theorem}
\fi

\section{Hardness Results}

\label{sec:nphardfixeddirections}
The $\bigplus$-\srect{}-$\msquare$ is \PSPACE-hard; this follows from the reduction of~\cite{hopcroft1984}, since the rectangles in hard instances of~\cite{hopcroft1984} move horizontally or vertically. Also, an instance, when feasible, is feasible by a serial motion. Therefore, $\bigplus$-\prect{}-$\msquare$ is also \PSPACE-hard. The restriction to axis-aligned motion actually makes the reduction in~\cite{hopcroft1984} simpler.




The following theorem shows that \srect{}, restricted to instances in which the set $\VVV$ of directions is a fixed set containing at least two nonparallel directions, is \NP-hard. From this and the results of Sections~\ref{sec:grid} and~\ref{sec:general}, it follows that the problem is \NP-complete.


\iflong
The hardness reduction follows the ideas of~\cite{calinescu2008,dumitrescu2013} for the case of unit-disk robots and unrestricted translation motion. Essentially, the same proof works for the problem where the rectangular robots can move in arbitrary directions. For moves along a fixed set of directions, however, we need to make some nontrivial modifications to the reduction.
\fi

\ifshort
\begin{theorem}
  [Appendix] 
\label{t:nphard}
\label{thm:nphard}
    \srect{} restricted to the set of instances in which $\VVV$ is fixed and contains two non-parallel directions, is \NP-complete.
\end{theorem}
\fi

\iflong
\begin{theorem}
   
\label{t:nphard}
\label{thm:nphard}
    \srect{} restricted to the set of instances in which $\VVV$ is fixed and contains two non-parallel directions, is \NP-complete.
\end{theorem}
\begin{proof}

It was shown in Corollary~\ref{cor:NP} that the problem is in \NP. Here we prove \NP-hardness.
Our starting point is the 3-SET-COVER problem. In the 3-SET-COVER problem, we have $m$ sets $S_1, \ldots, S_m \subset U$ where $|U|=n$. Given $k\geq 1$, the problem asks if there exist at most $k$ sets $S_i$  that cover $U$, that is, their union is $U$. We set $s_i:=|S_i|\leq 3$. As in~\cite{dumitrescu2013}, we consider the graph $G$ depicted in Figure~\ref{fig:drawing}. The nodes of $G$ are partitioned into three sets as in the figure. Those in $A$ are in a one-one correspondence with the elements of $U$.

The nodes in $B$ correspond to the sets $S_i$, and the nodes in part $C$ are again in one-one correspondence with elements of $U$.

We consider a fixed planar drawing of the graph $G$. In the original problem considered by \cite{dumitrescu2013}, a robot can travel along the edge $S_iu_j$ in one move. However, for us this is not true and here we need to adapt the technique. Since $\VVV$ contains two non-parallel directions, take any two such directions, say $\overrightarrow{e_1},\overrightarrow{e_2} \in \VVV$. The main idea is to draw the graph $G$ such that each edge is drawn as a (possibly non-uniform) zigzag using $\pm \overrightarrow{e_1}$ and $\pm \overrightarrow{e_2}$. To make this work, we need to make a careful adjustment of the number of zigs and zags on different parts of the edges $S_iu_j$. To show the \NP-hardness, it suffices to restrict ourselves to the case where the robots are unit squares.

\begin{figure}
    \centering
    \includegraphics{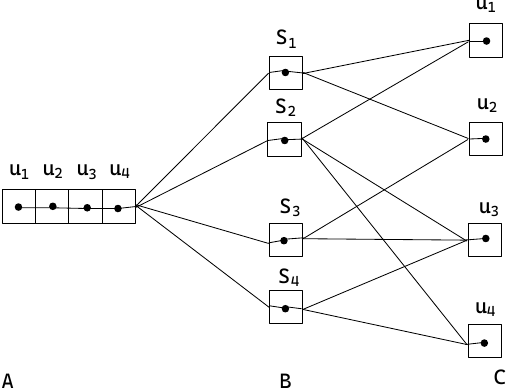}
    \caption{The graph drawing used in the reduction. Straight line segments will be replaced with zigzags. The squares depict the start and end positions of robots and are not part of the graph.}\label{fig:drawing}
    \label{fig:enter-label}
\end{figure}

In the drawing of Figure~\ref{fig:drawing}, intersections happen between arcs $S_iu_j$. By moving the nodes in $B$ by a small distance up or down we can make sure that these intersection points are on different $x$-coordinates. Imagine we separate these intersection points by vertical lines $l_1, \ldots, l_\chi$ placed between any two consecutive intersection points. These lines will cut all segments $S_iu_j$. There are at most $3m$ edges, and at most $9m^2$ intersections. The total number of segments that the edges $S_iu_j$ are cut into is $O(m^3)$. The vertical lines partition the strip between nodes $B$ and $C$ into $\chi+1$ smaller, interior-disjoint strips. We will consider a typical strip $\Sigma$. In the interior of $\Sigma$, there is a single intersection between arcs $S_iu_j$, see Figure~\ref{fig:zigzag}. Let $\Sigma_{i,j} = \sigma_{i,j}^\Sigma$ denote the segments $S_iu_j \cap \Sigma$. We replace each segment $\sigma_{i,j}$ by a broken line in the form of a zigzag, see Figure~\ref{fig:zigzag}. By making the steps in each segment small enough, we can make sure that the broken lines for disjoint segments are disjoint. Moreover, if two segments $\sigma$ and $\sigma'$ intersect inside $\Sigma$, by making the ``teeth'' of one broken line much smaller than the other, we can make sure there is no intersection between the broken lines other than the original intersection; see Figure~\ref{fig:zigzag}. Let $\nu_{i,j} = \nu_{i,j}^\Sigma$ be the number of steps in the staircase of $\sigma_{i,j}$. By injecting extra small steps into the staircase, we can make sure that all the broken lines have the same number of turns $\nu_\Sigma = \max_{i,j}{\nu_{i,j}}$, see Figure~\ref{fig:addteeth}. Moreover, we make sure that each broken line, from those replaced for $\sigma$ and $\sigma'$, has the same number of turns before the intersection, and after the intersection. This again can be achieved by inserting small steps. Therefore, it is possible to make all the turn numbers equal.

\begin{figure}
    \centering
    \includegraphics{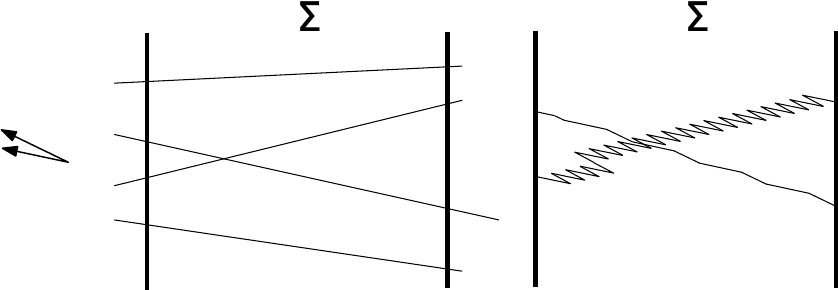}
    \caption{Left: two non-parallel directions, center and right: replacing segments with zigzags in the two given directions.}
    \label{fig:zigzag}
     
\end{figure}

\begin{figure}
    \centering
    \includegraphics{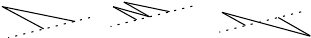}
    \caption{Adding extra teeth to zigzags.}
    \label{fig:addteeth}
\end{figure}
We perform the above operation for each strip $\Sigma$. We make sure that when two consecutive broken lines join on a vertical line $l_j$, no turn is generated. Further, again by injecting extra turns, we can make sure that the number of turns of any segment, inside any strip equals $\nu = \max_\Sigma{\nu_\Sigma}$. This finishes the new drawing of the part of the graph $G$ between $B$ and $C$. We then draw the nodes of part $A$ on a line of direction either $\overrightarrow{e_1}$ or $\overrightarrow{e_2}$. We then connect the last of these nodes, by broken lines with $\nu$ turns each, to all the nodes of $B$. It is possible to do this without creating intersections around the last node of $A$; see Fig~\ref{fig:anode}.  We can draw the legs of $C$ easily. This finishes the drawing of the graph. 

\begin{figure}
    \centering
    \includegraphics{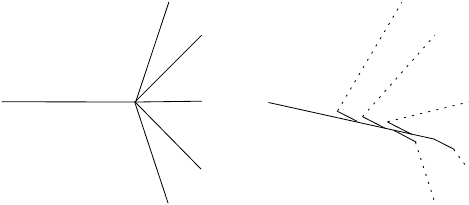}
    \caption{Realizing a high-degree node with two slopes. The dashed lines will be replaced with zigzags, each robot still goes through the same number of turns.}
    \label{fig:anode}
\end{figure}


The rest of the proof now again follows the ideas of \cite{dumitrescu2013}. Since the directions are constant, the finest geometrical details in the drawing are lower bounded by a function of form $1/p(m)$ for a polynomial $m$, and the number $\nu$ is upper bounded by a polynomial in $p$. Therefore, we can realize the (new) drawing of the graph inside a box of size $q(m)\times q(m)$, where $m$ is the number of sets and $q$ is a polynomial. The box is decomposed into unit squares. By scaling the box drawing by a polynomial factor, we can make sure that each arc of the drawing connecting two consecutive turns crosses a large enough number of cells of the grid.

Let $Q$ denote a unit square. We next thicken the drawing by taking the space resulting from placing a unit square at each point of the drawing. Let this thickening be $G_Q$. We remove all the cells of the box which intersect $G_Q$, making room for robots to move along $G_Q$ centered at $G$. We put $s$ unit square robots at the nodes of $A$, and $S$ robots at the nodes of $B$, and we perturb them a little so that they cannot move to the nodes in $C$ in one move (in our setting this is really needed in very degenerate situations). Each $S_i$, together, with all unit squares which we have not removed from the grid, have the same start and end positions. We also carve out a little path for each $S_i$ so that it can move temporarily to make room for movements of robots to their destinations in $C$. The whole graph will be put inside layers of extra obstacles around to simulate a bounding box. Polynomially-many of these would suffice.   

Each robot in $A$ takes $c + \nu + \nu(\chi+1)$ many moves at least to reach its destination in $C$. Here $c$ is a constant counting extra bends introduced around $A$ and $B$ and $C$. These are the same extra number of moves for all robots. The first $\nu$ is to get from $A$ to $B$, and then by $\nu(\chi+1)$ moves it can get from $B$ to $C$. The main point of the reduction is that there is $q$ sets that cover $U$, if and only $|U|(c+\nu+\nu(\chi+1))+2q$ moves are enough for moving all robots to their destinations. Here $2q$ is the number of moves made by $S_i$ to move into a temporary position and back. If a robot decides to make a turn at an intersection inside some $\Sigma$, it will make $\nu+1$ turns in that strip.

\end{proof}

We note that the \NP-completeness result in Theorem~\ref{t:nphard} does not contradict the \PSPACE-hardness of $\bigplus$-\srect{}-$\msquare$. The simulation of $\bigplus$-\srect{}-$\msquare$ by our problem requires simulating a bounding box. However, we can use only polynomially many robots, and $\bigplus$-\srect{}-$\msquare$ might need exponentially many moves.
\fi

\ifshort
\section{Concluding Remarks}
\label{sec:conclusion}
We studied the complexity and developed parameterized algorithms for fundamental computational geometry problems pertaining to the motion planning of rectangular robots in the plane. Several follow-up questions ensue from this work:  

\begin{enumerate}

\item 
What is the parameterized complexity of the problem variant in which there is no restriction on the translation directions? One possible approach to show \FPT{} for this variant is to show that there is a set of possible positions for the robots, that can be computed based on the instance, that depends on the geometric complexity polynomially, and that transforms the continuous problem into a discrete one. We conjecture this to be true in this case. 

\item What is the parameterized complexity of the problem for other geometric shapes (e.g., congruent disks)?  

\item What is the parameterized complexity of the problem for  environments  with obstacles? 
\end{enumerate}
\fi

\iflong
\section{Concluding Remarks}
\label{sec:conclusion}
We studied the complexity and developed parameterized algorithms for fundamental computational geometry problems pertaining to the motion planning of rectangular robots in the plane. We studied various settings of the problem, and showed them to be \NP-hard/complete and fixed-parameter tractable as well. 

The following follow-up questions to this research are important, natural to ask, and remain unresolved up to our knowledge:

The first question pertains to the parameterized complexity of the problem variant in which there is no restriction on the translation directions (i.e., rectangles can be translated in any direction).  More specifically, does the problem remain \FPT{} parameterized by the number of robots?  In essence, we are asking if there is a ``grid'', i.e. a set of possible positions for the robots, that can be computed beforehand, that depends on the geometric complexity polynomially, and that transforms the continuous problem into a discrete one. We conjecture this to be true for rational points when the translation directions are arbitrary.

The second question pertains to the fixed-parameter tractability of the translation of other geometric shapes, in particular the well-studied problem of translating disks (see Section~\ref{sec:intro}). The problem is known to be \NP-hard~\cite{dumitrescu2013}, and a modification of the reduction we presented in Section~\ref{sec:nphardfixeddirections} shows that the problem remains \NP-hard when restricted to axis-aligned translations. The membership of the problem in \NP{} is not known. The axis-aligned translation of (unit) disks seems to be computationally harder than that of rectangles. In particular, the techniques we used in Section~\ref{sec:general} to show fixed-parameter tractability do not seem to work for disks mainly due to our inability to model collision using linear constraints (even after enumeration), as modelling collision seems to require constraints that involve distances, and hence, may not be encoded using linear constraints.  

Finally, the problems under consideration in this paper become naturally harder when the robots move in an environment with obstacles. As observed in this paper, even with a bounding box (and no other obstacles) the problems become \PSPACE-hard. It is interesting to investigate the parameterized complexity of these problems when obstacles are present in the environment. In particular, do the problems for which we obtained \FPT{} algorithms in this paper remain \FPT{} (parameterized by $k$) in an environment with (polygonal obstacles)?
\fi

\bibliography{ref}

\end{document}